\documentclass[letterpaper, 11pt]{article}

\usepackage[top=1.0in, bottom=1.0in, left=1.0in, right=1.0in]{geometry}
\usepackage{amsmath, amsfonts, amsthm, amssymb}

\newcommand{\E}{\mathbb{E}}

\newcommand{\R}{\mathbb{R}}

\newcommand{\N}{\mathbb{N}}

\newcommand{\Var}{\text{Var}}

\newcommand{\wt}{\widetilde}

\newcommand{\eps}{\epsilon}

\renewcommand\vec[1]{\ensuremath\boldsymbol{#1}}
\renewcommand{\v}{\vec}
\renewcommand{\H}{\mathcal{H}}
\newcommand{\T}{\mathcal{T}}
\newcommand{\sw}{\mathsf{sw}}

\makeatletter
\newtheorem*{rep@theorem}{\rep@title}
\newcommand{\newreptheorem}[2]{%
\newenvironment{rep#1}[1]{%
  \def\rep@title{#2 \ref{##1}}%
  \begin{rep@theorem}}%
  {\end{rep@theorem}}}
\makeatother

\newtheorem{theorem}{Theorem}

\newtheorem{lemma}[theorem]{Lemma}
\newtheorem{claim}[theorem]{Claim}

\newtheorem{definition}[theorem]{Definition}

\newreptheorem{theorem}{Theorem}
\newreptheorem{proposition}{Proposition}
\newreptheorem{corollary}{Corollary}
\newreptheorem{lemma}{Lemma}
\newreptheorem{claim}{Claim}
\newreptheorem{conjecture}{Conjecture}

\theoremstyle{definition}

\newtheorem{example}[theorem]{Example}
\newtheorem*{remark}{Remark}

\begin{document}

\title{Bayesian Mechanism Design with Efficiency, Privacy, and Approximate Truthfulness\thanks{A preliminary version of this paper appeared in the 8th Workshop on Internet \& Network Economics (WINE 2012). This work was supported in part by NSF grants IIS-0534064, IIS-0812045, IIS-0911036, CCF-0746990; AFOSR grants FA9550-08-1-0438, FA9550-09-1-0266, FA9550-10-1-0093; ARO grant W911NF-09-1-0281. We thank Joseph Y. Halpern for helpful discussions.}}

\author{Samantha Leung and Edward Lui \\
Department of Computer Science, Cornell University \\
\texttt{\{samlyy,luied\}@cs.cornell.edu}}

\date{December 7, 2012}

\maketitle            

\begin{abstract}
Recently, there has been a number of papers relating mechanism design and privacy (e.g., see \cite{MT07,Xia11,CCKMV11,NST12,NOS12,HK12}). All of these papers consider a worst-case setting where there is no probabilistic information about the players' types. In this paper, we investigate mechanism design and privacy in the \emph{Bayesian} setting, where the players' types are drawn from some common distribution. We adapt the notion of \emph{differential privacy} to the Bayesian mechanism design setting, obtaining \emph{Bayesian differential privacy}. We also define a robust notion of approximate truthfulness for Bayesian mechanisms, which we call \emph{persistent approximate truthfulness}. We give several classes of mechanisms (e.g., social welfare mechanisms and histogram mechanisms) that achieve both Bayesian differential privacy and persistent approximate truthfulness. These classes of mechanisms can achieve optimal (economic) efficiency, and do not use any payments. We also demonstrate that by considering the above mechanisms in a modified mechanism design model, the above mechanisms can achieve actual truthfulness.
\end{abstract}

\section{Introduction}

One of the main goals in mechanism design is to design mechanisms that achieve a socially desirable outcome even if the players behave selfishly. Because of the revelation principle, mechanism design has focused on direct (revelation) mechanisms where each player simply reports his/her private type (or valuation). This leads to the issue of privacy, where the players may be concerned that the mechanism's output may leak information about their private types (even if the mechanism is trusted). 

\paragraph{{\bf Mechanism Design and Privacy.}} Traditional mechanism design did not include the aspect of privacy. However, in the context of releasing information from databases, the issue of privacy has already been studied quite extensively. In this context, the current standard notion of privacy is \emph{differential privacy} \cite{DMNS06,Dwo06}. A data release algorithm satisfies differential privacy if the algorithm's output distribution does not change much when one person's data is changed in the database. This implies that the algorithm does not leak much information about any person in the database. 

Recently, there has been a number of papers that combine mechanism design with differential privacy. In \cite{MT07}, McSherry and Talwar develop a general mechanism called the \emph{exponential mechanism} that is differentially private; they also show that any differentially private mechanism is \emph{approximately} truthful. In \cite{NST12}, Nissim, Smorodinsky, and Tennenholtz modify the standard mechanism design model by adding a ``reaction stage''; in this new model, the authors combine differentially private mechanisms with a ``punishing mechanism'' to obtain mechanisms that are \emph{actually} truthful. However, the mechanisms in \cite{NST12} might not protect the privacy of the players, due to the reaction stage. 

The main goal of the above two papers was to use differential privacy as a tool for achieving some form of truthfulness, as opposed to achieving privacy for the players. However, there has been other papers that focus on designing mechanisms that protect the privacy of the players. In \cite{HK12}, Huang and Kannan show that a pricing scheme can be added to the exponential mechanism to make it \emph{actually} truthful, resulting in a general mechanism that is both differentially private and truthful. In \cite{Xia11}, Xiao provides a transformation that takes truthful mechanisms and transforms them into truthful and differentially private mechanisms. On the other hand, Xiao also shows that a mechanism that is truthful and differentially private might not be truthful in a model where the players are ``privacy-aware'', i.e., privacy is explicitly captured in the players' utility functions. In \cite{CCKMV11}, Chen et al. construct mechanisms that are truthful even in a model where the players are privacy-aware. In \cite{NOS12}, Nissim, Orlandi, and Smorodinsky construct mechanisms that are truthful in a different privacy-aware model.

\paragraph{{\bf Bayesian Mechanism Design.}} One desirable property of a mechanism is (economic) \emph{efficiency}; in fact, it would be best if the mechanism always chooses a social alternative that is \emph{optimal} with respect to some measure of efficiency, such as social welfare. However, such \emph{optimal efficiency} is not achieved by any of the above results. In fact, it is not possible for a differentially private mechanism to achieve optimal efficiency (for a non-trivial problem), since the mechanism has to be randomized in order to satisfy differential privacy. However, all of the above results are in a worst-case setting where there is no probabilistic information about the players' types. If we consider a non-worst-case setting, then it may be possible for a mechanism to achieve differential privacy without using any randomization. 

One such setting is the \emph{Bayesian} setting, where the players' types are drawn from some distribution. Such a setting follows the Bayesian approach that has been the standard in economic theory for many decades. Recently, mechanism design in the Bayesian setting has also been gaining popularity in the computer science community. Thus, it is interesting to consider the issue of privacy in the Bayesian setting as well. In particular, it may be possible for a Bayesian mechanism to achieve optimal efficiency while satisfying some form of differential privacy. Achieving optimal efficiency may be critical for certain problems, such as presidential elections and kidney transplant allocations, where it may be unethical and/or unfair to make a non-optimal choice. Although differentially private mechanisms in the worst-case setting may asymptotically achieve nearly optimal efficiency in expectation (or with reasonably high probability), there is no guarantee that the chosen outcome for a particular execution of the mechanism is actually close to optimal. 

\paragraph{{\bf Bayesian Differential Privacy and Persistent Approximate Truthfulness.}} In this paper, we consider mechanism design in the Bayesian setting, and our main goal is to construct useful mechanisms that achieve optimal efficiency, some form of differential privacy, and some notion of truthfulness. Since differential privacy is a worst-case notion in the sense that no distributional assumptions are made on the input of the mechanism, we first adapt the notion of differential privacy to the Bayesian mechanism design setting. We call this new notion \emph{Bayesian differential privacy}; this is the privacy notion that we use in this paper. 

As mentioned above, Xiao \cite{Xia11} showed that a mechanism that is truthful and differentially private might not be truthful in a model where privacy is explicitly captured in the players' utility functions. In this paper, we do not use such a model, since there are many settings where the players would already be satisfied with differential privacy and would not report strategically in an attempt to further protect their privacy. Our results will be meaningful in these settings; furthermore, even in a setting where we want to explicitly capture privacy in the players' utility functions, our techniques and results can still be useful in constructing truthful mechanisms (similar to how the mechanisms in \cite{CCKMV11} and \cite{NOS12} are still based on differentially private mechanisms).

We also want our mechanisms to satisfy some form of truthfulness. The standard notion of truthfulness in Bayesian mechanism design is that the truthful strategy profile is a Bayes-Nash equilibrium. Similar to \cite{MT07}, we first relax truthfulness so that the truthful strategy profile only needs to be an \emph{$\eps$-Bayes-Nash equilibrium}, where an $\eps$ margin is allowed in the Nash conditions. However, we would like to obtain notions of truthfulness that are stronger than that provided by the $\eps$-Bayes-Nash equilibrium. Thus, we strengthen the $\eps$-Bayes-Nash equilibrium such that even if up to $k$ players deviate from the equilibrium, everyone else's best-response is still to adhere to their part of the equilibrium. We call this new equilibrium concept the \emph{$k$-tolerant $\eps$-Bayes-Nash equilibrium}. We would also like our equilibrium concept to be resilient against coalitions. Thus, we further strengthen our notion of $k$-tolerant $\eps$-Bayes-Nash equilibrium to \emph{$(k,r)$-persistent $\eps$-Bayes-Nash equilibrium}, which is resilient against coalitions of size $r$ even in the presence of $k$ deviating players. The notion of truthfulness we use requires that the truthful strategy profile is a $(k,r)$-persistent $\eps$-Bayes-Nash equilibrium, which we will refer to as \emph{persistent approximate truthfulness}.

\subsection{Our Results}

In this paper, we present three classes of mechanisms that achieve both Bayesian differential privacy and persistent approximate truthfulness:

\paragraph{{\bf Histogram Mechanisms.}} Roughly speaking, a histogram mechanism is a mechanism that first computes a histogram from the reported types, and then chooses a social alternative based only on the histogram. In Section \ref{sec:histogramMechanisms}, we show that if every bin of the histogram has positive expected count, then the histogram mechanism is both Bayesian differentially private and persistent approximately truthful. 

\paragraph{{\bf Mechanisms for Two Social Alternatives.}} Roughly speaking, this class includes any mechanism that makes a choice between two social alternatives $\{A,B\}$ based on the difference between the sums of two functions $u(\cdot, A)$ and $u(\cdot,B)$ on the types. In Section \ref{sec:twoChoicesMechanisms}, we show that as long as the random variable $u(t,A) - u(t,B)$ (where $t$ is distributed according to the type distribution) has non-zero variance and a pdf, then such a mechanism is both Bayesian differentially private and persistent approximately truthful.

\paragraph{{\bf Social Welfare Mechanisms.}} Roughly speaking, this class includes any mechanism that makes a choice based on the social welfare provided by each social alternative. An important subset of these mechanisms is the set of mechanisms that maximize social welfare. In Section \ref{sec:swMechanisms}, we show that if the players' valuations for each social alternative are normally distributed, then such a mechanism is both Bayesian differentially private and persistent approximately truthful. Also, we generalize this result to the case where the players' valuations for each social alternative are arbitrarily distributed with non-zero variance and a pdf.
\\ \\
The mechanisms in the above three classes are all deterministic and can achieve optimal efficiency. Furthermore, the mechanisms do not use any payments, and do not need to depend on a ``common prior'' or the players' beliefs. As long as the players' beliefs satisfy our distributional assumptions, our results on approximate truthfulness still hold. Moreover, as long as the input of the mechanism satisfies our distributional assumptions, the privacy of each player is protected.

\paragraph{{\bf Obtaining Actual Truthfulness.}} Recall that in \cite{NST12}, the authors added a ``reaction stage'' to the standard mechanism design model in order to achieve actual truthfulness from approximate truthfulness (which is obtained via differential privacy). We can also use this model and their techniques to obtain actual truthfulness in our results. In Section \ref{sec:punishing}, we describe an alternative model where actual truthfulness can be obtained from approximate truthfulness. In this new model, the mechanism is given the ability to verify the truthfulness of a small number of players. This model is simple to use and is realistic in settings where the truthfulness of a player can be verified objectively (e.g., income, expenses, age, address).

\section{Preliminaries and Definitions}
\label{sec:preliminaries}

For any $k \in \N$, we will use $[k]$ to denote the set $\{1, \ldots, k\}$. We consider a mechanism design environment consisting of the following components:
\begin{itemize}
\item A number $n$ of \emph{players}; we will often use $[n]$ to denote the set of $n$ players.
\item A \emph{type space} $T$; each player has a private type from the type space $T$. (We assume that each player has the \emph{same} type space $T$.)
\item A distribution $\T$ over the type space $T$; the players' private types are \emph{independently} drawn from this distribution. (For simplicity, we assume that the type for each player has the \emph{same} distribution $\T$; however, most of our results can be generalized to the case where the type for each player may follow a different distribution.)
\item A set $S$ of \emph{social alternatives}; for convenience, we assume that $S$ is finite.
\item For each player $i$, a \emph{utility function} $u_i: T \times S \to \R$; for $t \in T$ and $s \in S$, $u_i(t,s)$ represents the utility that player $i$ receives if player $i$ has type $t$ and the social alternative $s$ is chosen.
\end{itemize} 

We will focus on direct revelation mechanisms where each player reports his/her type.
Therefore, a \emph{mechanism} is a function $M: T^n \to S$, and a (pure) strategy for player $i$ is a function $\sigma_i: T \to T$ that maps true types to announced types. For convenience, whenever we refer to a mechanism $M: T^n \to S$, we assume that it is associated with an environment as described above. 

\begin{remark}
For simplicity of presentation, we have assumed in our model that there is a ``common prior'', i.e., it is common knowledge that the player types are distributed according to $\T^n$. However, most of our results can be easily generalized to a setting where each player has a \emph{different} belief about the other players' types (as long as the belief satisfies certain distributional assumptions, which is needed even in the common prior setting). Also, the mechanisms that we give do not need to depend on a common prior or the players' beliefs. As long as the players' beliefs satisfy our distributional assumptions, our results on approximate truthfulness still hold. Furthermore, as long as the input of the mechanism satisfies our distributional assumptions, the privacy of each player is protected by our results on privacy.
\end{remark}

\subsection{Equilibrium Concepts}

In this section, we will define several equilibrium concepts based on the standard Bayes-Nash equilibrium (see, e.g., \cite{FT91}).
These equilibrium concepts will be used to define various notions of truthfulness.
Our definitions build on the \emph{$\eps$-Bayes-Nash equilibrium}, which is a relaxation of the Bayes-Nash equilibrium in the sense that an $\eps$ margin is allowed in the Nash conditions.
This relaxation reflects the assumption that players will not deviate from the equilibrium if gains from deviation are sufficiently small.
In this paper, we also refer to $\eps$-Bayes-Nash equilibria as approximate Bayes-Nash equilibria. 
For more information about various notions of approximate equilibria, see \cite{NRTV07,SLB09,Tij81}.

Our equilibrium concepts strengthen the $\eps$-Bayes-Nash equilibrium.
We chose two strengthenings to address the following weaknesses of Nash equilibria. 
Firstly, a player's part of a Nash equilibrium is only guaranteed to be a best-response if all the other players are playing their parts of the equilibrium. 
In other words, a Nash equilibrium cannot tolerate players deviating from their equilibrium strategy --- if there is one irrational person in the system, the equilibrium breaks down. 
Deviations are especially problematic in $\eps$-equilibria, where there is less confidence that everyone would play their part of the equilibrium.
Secondly, a Nash equilibrium is not {resilient} to deviations by more than one person; coalitions of players can have profitable deviations from the equilibrium.

To address the first problem, we strengthen the Nash conditions such that even if up to $k$ players deviate from the equilibrium, everyone else's best-response is still to adhere to their part of the equilibrium.
In other words, the equilibrium \emph{tolerates} arbitrary deviations of $k$ individuals. 

\begin{definition}[{\bf $k$-tolerant $\epsilon$-Bayes-Nash equilibrium}]
A strategy profile $\vec{\sigma} = (\sigma_1, \ldots, \sigma_n)$ is a \emph{$k$-tolerant $\epsilon$-Bayes-Nash equilibrium} if for every $I \subseteq [n]$ with $|I| \leq k$, every possible announced types $\vec{t}'_I \in T^{|I|}$ for $I$, every player $i \notin I$, and every pair of types $t_i, t_i'$ for player $i$, we have 
\begin{align*}
\E_{\vec{t}_{J}}[u_i(t_i,M(\sigma_i(t_i),\vec{t}'_I,\vec{\sigma}_J(\vec{t}_J)))] \geq \E_{\vec{t}_{J}}[u_i(t_i,M(t_i', \vec{t}'_I, \vec{\sigma}_J(\vec{t}_J)))] - \epsilon,
\end{align*}
where $J = [n] \setminus (I \cup \{i\})$ and $\vec{t}_J \sim \T^{|J|}$.
\end{definition}

We note that $k$-tolerance is distinct from the notion of \emph{$k$-immunity} as defined in \cite{ADHG06,Hal08}, which guarantees that when up to $k$ people deviate from the equilibrium, the utilities of the non-deviating players do not decrease. 

The second problem mentioned above is addressed by \emph{$r$-resilience} (see, e.g., \cite{NRTV07,Hal08}).
A Bayes-Nash equilibrium is $r$-resilient if for any group of size at most $r$, there does not exist a deviation of the group such that any member of the coalition has increased utility. 

% Definition of $r$ resilience goes here 
\begin{definition}[{\bf $r$-resilient $\epsilon$-Bayes-Nash equilibrium}]
A strategy profile $\vec{\sigma} = (\sigma_1, \ldots, \sigma_n)$ is an \emph{$r$-resilient $\epsilon$-Bayes-Nash equilibrium} if for every coalition $C \subseteq [n]$ with $|C| \leq r$, every true types $\vec{t}_C \in T^{|C|}$ for $C$, every player $i \in C$, and every possible announced types $\vec{t}'_C \in T^{|C|}$ for $C$, we have 
\begin{align*}
\E_{\vec{t}_{-C}}[u_i(t_i,M(\vec{\sigma}_C(\vec{t}_C), \vec{\sigma}_{-C}(\vec{t}_{-C})))] \geq \E_{\vec{t}_{-C}}[u_i(t_i,M(\vec{t}'_C,\vec{\sigma}_{-C}(\vec{t}_{-C})))] - \epsilon,
\end{align*}
where $\vec{t}_{-C} \sim \T^{n-|C|}$.
\end{definition}

% explain that tolerance and resilience are different
It is not hard to see that resilience and tolerance can be independently violated, and hence neither implies the other.
Just as the authors in \cite{ADHG06,Hal08} combine immunity and resilience, we consider the combination of tolerance and resilience.
Roughly speaking, a \emph{$(k,r)$-persistent Bayes-Nash equilibrium} is a Bayes-Nash equilibrium that is $r$-resilient (protects against coalitions of size $r$), even in the presence of up to $k$ individuals that are deviating arbitrarily from the equilibrium.

% Definition of persistent equilibrium goes here , next talk about truthfulness
\begin{definition}[{\bf $(k,r)$-persistent $\epsilon$-Bayes-Nash equilibrium}]
A strategy profile $\vec{\sigma} = (\sigma_1, \ldots, \sigma_n)$ is a \emph{$(k,r)$-persistent $\epsilon$-Bayes-Nash equilibrium} if for every $I \subseteq [n]$ with $|I| \leq k$, every possible announced types $\vec{t}'_I \in T^{|I|}$ for $I$, every coalition $C \subseteq [n] \setminus I$ with $|C| \leq r$, every true types $\vec{t}_C \in T^{|C|}$ for $C$, every player $i \in C$, and every possible announced types $\vec{t}'_C \in T^{|C|}$ for $C$, we have 
\begin{align*}
\E_{\vec{t}_J}[u_i(t_i,M(\vec{\sigma}_C(\vec{t}_C), \vec{t}'_I, \vec{\sigma}_{J}(\vec{t}_{J})))] \geq \E_{\vec{t}_{J}}[u_i(t_i,M(\vec{t}'_C, \vec{t}'_I, \vec{\sigma}_{J}(\vec{t}_{J})))] - \epsilon,
\end{align*}
where $J = [n] \setminus (I \cup C)$ and $\vec{t}_J \sim \T^{|J|}$.
\end{definition}
%Clearly, a $(k,r)$-persistent $\epsilon$-Bayes-Nash equilibrium is also $(k+r-1)$-tolerant and $r$-resilient.

\subsection{Notions of Truthfulness}
\label{sec:truthfulness}

In this section, we define various notions of truthfulness based on the equilibrium concepts from the previous section. 
Recall that a mechanism is \emph{Bayes-Nash truthful} if the truthful strategy profile is a Bayes-Nash equilibrium.
Similarly, a mechanism is \emph{$\eps$-Bayes-Nash truthful} if the truthful strategy profile is an $\eps$-Bayes-Nash equilibrium.
By using the equilibrium concepts from the previous section, we can obtain stronger notions of truthfulness. 

% Definition of truthfulness 
\begin{definition}[{\bf [$k$-tolerant]/[$r$-resilient]/[$(k,r)$-persistent] $\eps$-Bayes-Nash truthful}]\label{def:truthfulness}
A mechanism is \emph{$k$-tolerant $\eps$-Bayes-Nash truthful} if the truthful strategy profile is a $k$-tolerant $\eps$-Bayes-Nash equilibrium. Similarly, a mechanism is \emph{$r$-resilient} (resp., \emph{$(k,r)$-persistent}) \emph{$\eps$-Bayes-Nash truthful}  if the truthful strategy profile is an $r$-resilient (resp., $(k,r)$-persistent) $\eps$-Bayes-Nash equilibrium. 
\end{definition}

It is easy to see that if a mechanism is $(k,r)$-persistent $\eps$-Bayes-Nash truthful, then it is also $k$-tolerant $\eps$-Bayes-Nash truthful and $r$-resilient $\eps$-Bayes-Nash truthful. In many settings, it is reasonable to believe that players in an $\eps$-Bayes-Nash truthful mechanism will be truthful, since (1) truth-telling is simple while computing a profitable deviation can be costly (see, e.g., \cite{HP10}), and (2) lying can induce a psychological (morality) cost.
Indeed, there are many results in mechanism design that assume that approximate truthfulness is enough to ensure that players will be truthful (see, e.g., \cite{MT07,BP11,Kot04,Sch04}). 

\section{Privacy for Bayesian Mechanism Design}

In this section, we describe and define \emph{Bayesian differential privacy}, which is a natural adaptation of \emph{differential privacy} \cite{DMNS06,Dwo06} to the Bayesian mechanism design setting. Roughly speaking, differential privacy requires that when one person's input to the mechanism is changed, the output distribution of the mechanism changes very little (here, the mechanism is randomized). 

We now describe Bayesian differential privacy. We first note that even though the players' true types are drawn from some distribution $\T$, if all the players are non-truthful and announce a type independently of their true type, then the input of the mechanism is no longer distributional and we are essentially in the same scenario as in (worst-case) differential privacy. Thus, it is necessary to make some assumptions on the strategies of the players, so that the input of the mechanism contains at least some randomness. 

In our notion of Bayesian differential privacy, we assume that at least some players (e.g., a constant fraction of the players) are truthful so that their announced types have the same distribution as their true types. This assumption is not unreasonable, since we later show that if a mechanism is Bayesian differentially private, then the mechanism is automatically persistent approximately truthful, so we expect that most players would be truthful anyway. In particular, if we have an equilibrium where most players are truthful, then privacy is achieved at this equilibrium.

Roughly speaking, \emph{$(k,\eps,\delta)$-Bayesian differentially privacy} requires that when a player $i$ changes his/her announced type, the output distribution of the mechanism changes by at most an $(\eps,\delta)$ amount, assuming that at most $k$ players are non-truthful (possibly lying in an arbitrary way). This implies that the mechanism leaks very little information about each player's announced type, so each player's privacy is protected. The mechanism is assumed to be deterministic, so the randomness of the output is from the randomness of the types of the truthful players. (One can also consider randomized mechanisms, but we chose to focus on deterministic mechanisms in this paper.)

\begin{definition}[{\bf $(k,\eps,\delta)$-Bayesian differential privacy}]
A mechanism $M: T^n \to S$ is \emph{$(k,\eps,\delta)$-Bayesian differentially private} if for every player $i \in [n]$, every subset $I \subseteq [n] \setminus \{i\}$ of players with $|I| \leq k$, every pair of types $t_i, t_i' \in T$ for player $i$, and every $\vec{t}'_I \in T^{|I|}$, the following holds: Let $J = [n] \setminus (I \cup \{i\})$ (the remaining players) and $\vec{t}_J \sim \T^{|J|}$; then, for every $Y \subseteq S$, we have
\begin{align*}
\Pr[M(t_i, \vec{t}'_I, \vec{t}_J) \in Y] \leq e^\epsilon \cdot \Pr[M(t_i', \vec{t}'_I, \vec{t}_J) \in Y] + \delta,
\end{align*}
where the probabilities are over $\vec{t}_J \sim \T^{|J|}$.
\end{definition}

The parameter $k$ controls how many non-truthful players the mechanism can tolerate while satisfying privacy; $k$ can be a function of $n$ (the number of players), such as $k = \frac{n}{2}$. One can even view the non-truthful players as being controlled/known by an adversary that is trying to learn information about a player $i$'s type; as long as the adversary controls/knows at most $k$ people, player $i$'s privacy is still protected. The parameters $\eps$ and $\delta$ bound the amount of information about each person's (announced) type that can be ``leaked'' by the mechanism. Since the above definition of Bayesian differential privacy is a natural adaptation of differential privacy to Bayesian mechanism design, and since differential privacy is a well-motivated and well-accepted notion of privacy, we will not further elaborate on the details of the above definition. 

Our definition of $(k,\eps,\delta)$-Bayesian differential privacy has some similarities to the notion of \emph{$(\eps,\delta)$-noiseless privacy} (for databases) introduced and studied in \cite{BBGLT11}. However, there are some subtle but significant differences between the two definitions, so the results in this paper do not follow from the theorems and proofs in \cite{BBGLT11}. Nevertheless, the ideas and techniques in \cite{BBGLT11}, and for $(\eps,\delta)$-noiseless privacy in general, may be useful for designing Bayesian differentially private mechanisms. 

It is known that differentially private mechanisms are approximately (dominant-strategy) truthful (see \cite{MT07}). Similarly, Bayesian differentially private mechanisms are persistent approximate Bayes-Nash truthful.

\begin{theorem}[{\bf Bayesian differential privacy $\implies$ persistent approximate truthfulness}]
\label{thm:BDP=>Truthfulness}
Suppose the utility functions are bounded by $\alpha > 0$, i.e., the utility function for each player $i$ is $u_i: T \times S \to [-\alpha,\alpha]$. Let $M$ be any mechanism that is $(k,\eps,\delta)$-Bayesian differentially private. Then, $M$ satisfies the following properties:
\begin{enumerate}
  \item $M$ is $k$-tolerant $(\eps + 2\delta)(2\alpha)$-Bayes-Nash truthful.
  \item For every $1 \leq r \leq k+1$, $M$ is $r$-resilient $(r\eps + 2r\delta)(2\alpha)$-Bayes-Nash truthful.
  \item For every $1 \leq r \leq k+1$, $M$ is $(k-r+1,r)$-persistent $(r\eps + 2r\delta)(2\alpha)$-Bayes-Nash truthful.
\end{enumerate}
\end{theorem}

See \ref{app:privacyTruthfulness} for the proof of Theorem \ref{thm:BDP=>Truthfulness}.

\section{Efficient Bayesian Mechanisms with Privacy and Persistent Approximate Truthfulness\label{sec:mechs}}

In this section, we present three classes of mechanisms that achieve both Bayesian differential privacy and persistent approximate truthfulness. 

\subsection{Histogram Mechanisms}
\label{sec:histogramMechanisms}
We first present a broad class of mechanisms, called \emph{histogram mechanisms}, that achieve Bayesian differential privacy and persistent approximate truthfulness. 
Given a partition $P = \{B_1, \ldots, B_m\}$ of the type space $T$ with $m$ blocks (ordered in some way), a \emph{histogram} with respect to $P$ is simply a vector in $(\mathbb{Z}_{\geq 0})^m$ that specifies a count for each block of the partition. Given a partition $P$, let $\H_P$ denote the set of all histograms with respect to $P$; given a vector $\v{t}$ of types, let $H_P(\v{t})$ be the histogram formed from $\v{t}$ by simply counting how many components (types) of $\vec{t}$ belong to each block of the partition $P$.

We now define what we mean by a \emph{histogram mechanism}. Intuitively, a histogram mechanism is a mechanism that, on input a vector of types, computes the histogram from the types with respect to some partition $P$, and then applies any function $f: \H_P \to S$ to the histogram to choose a social alternative in $S$.

\begin{definition}[{\bf Histogram mechanism}]
Let $P$ be any partition of the type space $T$. A mechanism $M: T^n \to S$ is a \emph{histogram mechanism with respect to $P$} if there exists a function $f: \H_P \to S$ such that $M(\v{t}) = f(H_P(\v{t}))$ $\forall \ \v{t} \in T^n$.
\end{definition}

The following theorem states that any histogram mechanism with bounded utility functions and positive expected count for each bin is both Bayesian differentially private and persistent approximately truthful.

\begin{theorem}[{\bf Histogram mechanisms are private and persistent approximately truthful}\label{thm:histogramMechanism}]
Let $M: T^n \to S$ be any histogram mechanism with respect to some partition $P$ of $T$. Let $p_{\min} = \min_{B \in P} \Pr_{t \sim \T}[t \in B]$, and suppose that $p_{\min} > 0$. Then, for every $0 \leq k \leq n-2$ and $\frac{4}{p_{\min} \cdot (n-k-1)} \leq \eps \leq 1$, $M$ satisfies the following properties with $\delta = e^{-\Omega((n-k) \cdot p_{\min} \cdot \eps^2)}$:
\begin{enumerate}
  \item \emph{Privacy:} $M$ is $(k,\eps,\delta)$-Bayesian differentially private.
  \item \emph{Persistent approximate truthfulness:} Suppose the utility functions are bounded by $\alpha > 0$, i.e., the utility function for each player $i$ is $u_i: T \times S \to [-\alpha,\alpha]$. 
Then, for every $1 \leq r \leq k+1$, $M$ is $(k-r+1,r)$-persistent $(r\eps + 2r\delta)(2\alpha)$-Bayes-Nash truthful.
\end{enumerate}
\end{theorem}

The proof idea \emph{roughly} works as follows. Persistent approximate truthfulness follows from Bayesian differential privacy and Theorem \ref{thm:BDP=>Truthfulness}, so we only have to show that the histogram mechanism is Bayesian differentially private. Since the histogram mechanism chooses a social alternative based only on the computed histogram, it suffices to show that the computed histogram is Bayesian differentially private. Now, consider a player $i$ changing his/her announced type from $t_i$ to $t_i'$, and let $B_i$ and $B_i'$ be the two bins that contain $t_i$ and $t_i'$, respectively. When player $i$ changes from $t_i$ to $t_i'$, the count for bin $B_i$ decreases by 1, and the count for bin $B_i'$ increases by 1. However, the randomness of the other players' types can easily blur/smooth out this change. 

The histogram formed by the other players' types follows a multinomial distribution, which is relatively smooth near its expectation, i.e., the probability masses of two nearby histograms are relatively close to each other, provided that the two histograms are near the expectation. Now, we observe that by Chernoff bounds, the formed histogram will be near the expectation with high probability. Together, these facts imply that the distribution of the computed histogram when player $i$ reports $t_i$ is relatively close to the distribution when player $i$ reports $t_i'$. This shows that the computed histogram is Bayesian differentially private.

\begin{proof}
The second property follows from the first property and Theorem \ref{thm:BDP=>Truthfulness}. We now show the first property.

Let $i \in [n]$, let $I \subseteq [n] \setminus \{i\}$ with $|I| \leq k$, let $t_i, t_i' \in T$, let $\vec{t}'_I \in T^{|I|}$, and let $Y \subseteq S$. Let $J = [n] \setminus (I \cup \{i\})$ and $\vec{t}_J \sim \T^{|J|}$. We need to show that
\begin{align*}
  \Pr[M(t_i, \vec{t}'_I, \vec{t}_J) \in Y] \leq e^\epsilon \Pr[M(t_i', \vec{t}'_I, \vec{t}_J) \in Y] + \delta,
\end{align*}
which is equivalent to 
\begin{align*}
  \Pr[f(H_P(t_i, \vec{t}'_I, \vec{t}_J)) \in Y] \leq e^\epsilon \Pr[f(H_P(t_i', \vec{t}'_I, \vec{t}_J)) \in Y]
\end{align*}
for some function $f: \H_P \to S$. To show this, it is easy to see that it suffices to show that for every $W \subseteq \H_P$,
\begin{align*}
\Pr[H_P(t_i, \vec{t}_J) \in W] \leq e^\epsilon \Pr[H_P(t_i', \vec{t}_J) \in W] + \delta. \tag{1}
\end{align*}
To this end, fix $W \subseteq \H_P$. Let the partition $P = \{P_1, \ldots, P_m\}$, and let $p_j = \Pr_{t \sim \T}[t \in P_j]$. Let $b$ be the index $j$ of the partition $P_j$ that contains $t_i$, and let $b'$ be the index $j$ of the partition $P_j$ that contains $t_i'$. We can assume that $b \neq b'$, since otherwise (1) would hold trivially. Given a histogram $\vec{w}$, let $w_j$ be the count of the bin $P_j$ in $\vec{w}$. Let $n' = |J|$. Let $B$ be the set of all histograms $\vec{w}$ of size $n'+1$ (i.e., sum of the counts is $n'+1$) such that $w_b > p_b n' \cdot e^{\eps/2}$ or $w_{b'} < p_{b'} n' \cdot e^{-\eps/2}$. Now, observe that
\begin{align*}
\Pr[H_P(t_i,\vec{t}_J) \in W] 
= \ & \Pr[H_P(t_i,\vec{t}_J) \in W \cap \overline{B}] + \Pr[H_P(t_i,\vec{t}_J) \in W \cap B] \\
\leq \ & \Pr[H_P(t_i,\vec{t}_J) \in W \cap \overline{B}] + \Pr[H_P(t_i,\vec{t}_J) \in B]. \tag{2}
\end{align*}

\begin{claim}
For each $\vec{w} \in W \cap \overline{B}$,
\begin{align*}
\Pr[H_P(t_i, \vec{t}_J) = \vec{w}] \leq e^\epsilon \Pr[H_P(t_i', \vec{t}_J) = \vec{w}].
\end{align*}
\end{claim}

\begin{proof}[Proof of claim]
Fix $\vec{w} \in W \cap \overline{B}$. From the pmf of the multinomial distribution, we have
\begin{align*}
\frac{\Pr[H_P(t_i, \vec{t}_J) = \vec{w}]}{\Pr[H_P(t_i', \vec{t}_J) = \vec{w}]} = \frac{w_b}{w_{b'}} \cdot \frac{p_{b'}}{p_b} = \frac{w_b}{p_b n'} \cdot \frac{p_{b'} n'}{w_{b'}} \leq e^{\eps/2} \cdot e^{\eps/2} = e^\eps,
\end{align*}
where the last inequality follows from the fact that $w \notin B$.
\end{proof}

\begin{claim}
$\Pr[H_P(t_i,\vec{t}_J) \in B] \leq \delta$.
\end{claim}

\begin{proof}[Proof of claim]
Observe that
\begin{align*}
\Pr[H_P(t_i,\vec{t}_J) \in B] 
= \ & \Pr[H_P(t_i,\vec{t}_J)_b > p_b n' e^{\eps/2} \text{ or } H_P(t_i,\vec{t}_J)_{b'} < p_{b'} n' e^{-\eps/2}] \\
\leq \ & \Pr[H_P(t_i,\vec{t}_J)_b > p_b n' e^{\eps/2}] + \Pr[H_P(t_i,\vec{t}_J)_{b'} < p_{b'} n' e^{-\eps/2}] \\
= \ & \Pr[1 + Bin(n',p_b) > p_b n' e^{\eps/2}] + \Pr[Bin(n',p_{b'}) < p_{b'} n' e^{-\eps/2}] \\
= \ & \Pr[Bin(n',p_b) > p_b n' (e^{\eps/2} - (p_b n')^{-1})] + \Pr[Bin(n',p_{b'}) < p_{b'} n' e^{-\eps/2}] \\
\leq \ & \exp(- \Omega(p_b n' \cdot (e^{\eps/2} - (p_b n')^{-1} - 1)^2)) + \exp(- \Omega(p_{b'} n' \cdot (1-e^{-\eps/2})^2)) \\
\leq \ & \exp(- \Omega(p_b n' \cdot \eps^2)) + \exp(- \Omega(p_{b'} n' \cdot \eps^2)) \\
\leq \ & \exp(- \Omega(p_{\min} \cdot (n-k) \cdot \eps^2)),
\end{align*}
where $Bin(n',p)$ is a binomial random variable with $n'$ trials and success probability $p$, and the second inequality follows from Chernoff bounds.
\end{proof}

Now, by the above two claims, we have from (2) that
\begin{align*}
\Pr[H_P(t_i,\vec{t}_J) \in W]
\leq \ & \Pr[H_P(t_i,\vec{t}_J) \in W \cap \overline{B}] + \Pr[H_P(t_i,\vec{t}_J) \in B] \\
\leq \ & e^\eps \Pr[H_P(t_i',\vec{t}_J) \in W \cap \overline{B}] + \delta \\
\leq \ & e^\eps \Pr[H_P(t_i',\vec{t}_J) \in W] + \delta.
\end{align*}
This completes the proof.
\end{proof}

One possible partition of the type space is the one where there is a distinct block for each type. 
Thus, Theorem~\ref{thm:histogramMechanism} covers the case where the choice of the mechanism depends only on the \emph{number} of players that reported each type, and not their identities.
In fact, given any partition, one can redefine the type space so that the new types are the blocks of the partition. 
This means we could always redefine the type space and simply use the partition where there is a distinct block for each type in the new type space.
However, we believe it is more natural to preserve the original, natural type space, and to allow the histogram mechanism to use an appropriate partition of the type space.

%examples
In Theorem~\ref{thm:histogramMechanism}, since the histogram mechanism is not modified in any way to satisfy privacy and persistent approximate truthfulness, all properties of the mechanism (e.g., efficiency, truthfulness, individual rationality, etc.) are preserved. We now give a simple example to illustrate Theorem~\ref{thm:histogramMechanism}.

\begin{example}[{\bf Voting with Multiple Candidates}]\label{ex:voting}
Suppose we are trying to select a winner from a finite set of candidates (e.g., political candidates) using the plurality rule (i.e., each voter casts one vote and the candidate with the most votes wins). 
The set of social alternatives is the set of candidates, and the natural type space is the set of all preference orders over the candidates.
However, we can partition the type space such that each block $b$ represents a candidate $c_b$, and all the types with $c_b$ as their top choice belong to block $b$.
Intuitively, announcing a type that belongs to block $b$ can be understood as casting a vote for candidate $c_b$.
Using this partition, we can define a histogram mechanism that implements the plurality rule. We call this mechanism the \emph{plurality mechanism}.

It is well known that the plurality rule is not strategy-proof when there are more than two candidates (see, e.g., \cite{SLB09}). 
However, by Theorem~\ref{thm:histogramMechanism}, the plurality mechanism is Bayesian differentially private and persistent approximate Bayes-Nash truthful. For example, if $n$ is the number of voters, and each candidate is expected to get at least some constant fraction of the votes, then from  Theorem~\ref{thm:histogramMechanism} we can conclude that the plurality mechanism is $(\Omega(n),n^{-2/5}, e^{-\Omega(\sqrt[5]{n})})$-Bayesian differentially private and $(\Omega(n),\Omega(\sqrt[5]{n}))$-persistent $O(\frac{1}{\sqrt[5]{n}})$-Bayes-Nash truthful.
This means that even if a constant fraction of the voters lie arbitrarily, and coalitions of size $O(\sqrt[5]{n})$ can be formed, the privacy of each voter is still protected, and any possible gain from lying is bounded by $O(\frac{1}{\sqrt[5]{n}})$.
As $n\rightarrow\infty$, any possible gain from lying converges to $0$, making the plurality mechanism persistent Bayes-Nash truthful.
\end{example}

\begin{example}[{\bf Discrete Facility Location in Two Dimensions with Multiple Facilities}]\label{ex:facility}
Facility location has been given much attention in the literature on privacy in mechanism design (see, e.g., \cite{CCKMV11,Xia11,NST12}).
The standard facility location problem involves placing $m$ facilities (e.g. public libraries) within a region, so the set of social alternatives is the set of all ways to place the $m$ facilities.
Each player's type is his location within the region, and his utility is determined by the distance from his location to the nearest facility. 
Usually, the goal is to minimize the sum of the players' distances to their nearest facility. 

If we discretize the type space (i.e., the region), perhaps such that each city block is a type, we get a histogram of population counts in each city block.
Using a histogram mechanism, the mechanism designer can use any criteria to select the facility locations while achieving Bayesian differential privacy and persistent approximate truthfulness.
%%This means that the players need not worry about their locations being leaked, and that no one has much incentive to lie.
%Moreover, the possible facility locations are arbitrary and are not restricted by the discretization of the bins. 
%%For example, a gym company wants to place a new location within a $100\times 100$ block district to serve its $n$ VIP members. 
%%It is estimated that $\Omega(n)$ VIP members reside in each block. 
%%The company asks each VIP member in this district for his or her address, and then picks a location that minimizes the total distance from the member's block to the gym. 
%%By Theorem~\ref{thm:histogramMechanism}, this mechanism is Bayesian differentially private. 
%%Moreover, even if half of the VIP members report an arbitrary address, and coalitions of size $O(n^{1/4})$ can be formed, the possible gains of deviation is still bounded by $O(e^{-n})$, tending $0$ as $n\rightarrow\infty$. 
\end{example}

\subsection{Mechanisms for Two Social Alternatives}
\label{sec:twoChoicesMechanisms}

Although histogram mechanisms are useful in many settings, in order to apply Theorem~\ref{thm:histogramMechanism} to get good parameters, the number of bins cannot be extremely large. 
We now present a class of mechanisms that do not require the partitioning of types into bins, but are still Bayesian differentially private and persistent approximately truthful.
Roughly speaking, the following theorem states that any mechanism that makes a choice between two social alternatives $\{A,B\}$ based on the difference between the sums of two functions $u(\cdot, A)$ and $u(\cdot,B)$ on the types is Bayesian differentially private and persistent approximately truthful.

%technical stuff (theorem) goes here
\begin{theorem}[{\bf Private and persistent approximately truthful mechanisms for two social alternatives}]
\label{thm:twoChoicesMechanism}
Let $S = \{A, B\}$ be any set of two social alternatives, let $T$ be the type space, let $\T$ be any distribution over $T$, and let $u: T \times S \to [-\beta,\beta]$ be any function (e.g., a utility function for all the players). Suppose the random variable $u(t,A) - u(t,B)$, where $t \sim \T$, has non-zero variance and a probability density function. 

Let $M: T^n \to S$ be any mechanism such that 
\begin{align*}
M(\vec{t}) = f\left(\sum_{i=1}^n u(t_i, A) - \sum_{i=1}^n u(t_i,B)\right)
\end{align*}
for some function $f: \R \to S$. Then, for every $0 \leq k \leq n-2$ and $0 < \eps \leq 1$, $M$ satisfies the following properties
with $\eps' = \eps + O(\sqrt{\frac{{\ln (n-k)}}{{n-k}}} )$ and $\delta = O(\frac{1}{\eps \sqrt{n-k}})$:
\begin{enumerate}
  \item \emph{Privacy:} $M$ is $(k,\eps',\delta)$-Bayesian differentially private.
  \item \emph{Persistent approximate truthfulness:} Suppose the utility functions are bounded by $\alpha > 0$, i.e., the utility function for each player $i$ is $u_i: T \times S \to [-\alpha,\alpha]$. 
Then, for every $1 \leq r \leq k+1$, $M$ is $(k-r+1,r)$-persistent $(r\eps' + 2r\delta)(2\alpha)$-Bayes-Nash truthful.
\end{enumerate}
\end{theorem}

The proof idea \emph{roughly} works as follows. Persistent approximate truthfulness follows from Bayesian differential privacy and Theorem \ref{thm:BDP=>Truthfulness}, so we only have to show that the mechanism $M$ is Bayesian differentially private. Since the mechanism chooses a social alternative based only on $\sum_{i=1}^n u(t_i, A) - \sum_{i=1}^n u(t_i,B) = \sum_{i=1}^n (u(t_i, A)-u(t_i,B))$, it suffices to show that this sum is Bayesian differentially private. Now, consider a player $i$ changing his/her announced type from $t_i$ to $t_i'$. This changes the sum by at most $4\beta$. However, the randomness of the other players' types can easily blur/smooth out this change. 

After shifting and rescaling the terms of the sum appropriately, as $n \rightarrow \infty$, the pdf of the sum obtained from the other players' types converges uniformly to the pdf of the standard normal distribution (this follows from various ``local limit theorems''; e.g., see Statement 5 in Section 4 of Chapter VII in \cite{Pet75}). Thus, the standard normal distribution is a good approximation of the distribution of the sum obtained from the other players' types. Now, we note that the pdf of the standard normal distribution is relatively smooth near its expectation 0, i.e., the probability density at two nearby points is relatively close to each other, provided that the two points are near the expectation 0. Now, we observe that a sample from the standard normal distribution will be near its expectation 0 with high probability. Together, these facts imply that the distribution of the computed sum when player $i$ reports $t_i$ is relatively close to the distribution when player $i$ reports $t_i'$. This shows that the computed sum is Bayesian differentially private.

\begin{proof} 
The second property follows from the first property and Theorem \ref{thm:BDP=>Truthfulness}. We now show the first property.

Let $i \in [n]$, let $I \subseteq [n] \setminus \{i\}$ with $|I| \leq k$, let $t_i, t_i' \in T$, let $\vec{t}'_I \in T^{|I|}$, and let $Y \subseteq S$. Let $J = [n] \setminus (I \cup \{i\})$, $\vec{t}_J \sim \T^{|J|}$, and $n' = |J|$. Given a type $t \in T$, let $g(t) = u(t,A) - u(t,B)$. Given a vector $\vec{v} = (v_1, \ldots, v_m)$ of types in $T$, let $U(\vec{v}) = \sum_{j = 1}^{m} g(v_j)$. We need to show that
\begin{align*}
  \Pr[M(t_i, \vec{t}'_I, \vec{t}_J) \in Y] \leq e^{\epsilon'} \Pr[M(t_i', \vec{t}'_I, \vec{t}_J) \in Y] + \delta,
\end{align*}
which is equivalent to 
\begin{align*}
  \Pr[f\left(U(t_i,\vec{t}'_I, \vec{t}_J)\right) \in Y] \leq e^{\epsilon'} \Pr[f\left(U(t_i',\vec{t}'_I, \vec{t}_J)\right) \in Y] + \delta \tag{1}
\end{align*}
for some function $f: \R \to S$. Let 
\begin{align*}
h(t) = \frac{g(t) - \E_{t' \sim \T}[g(t')]}{(\Var_{t' \sim \T}[g(t')])^{1/2}},
\end{align*} 
and given a vector $\vec{v} = (v_1, \ldots, v_m)$ of types in $T$, let 
\begin{align*}
V(\vec{v}) = \sum_{j = 1}^{m} h(v_j). 
\end{align*}
We note that $\E_{t' \sim \T}[h(t')] = 0$ and $\Var_{t' \sim \T}[h(t')] = 1$. To show (1), it is easy to see that it suffices to show that for every measurable set $W \subseteq \R$,
\begin{align*}
\Pr\left[\frac{V(t_i, \vec{t}_J)}{\sqrt{n'}} \in W\right] \leq e^{\epsilon'} \Pr\left[\frac{V(t_i', \vec{t}_J)}{\sqrt{n'}} \in W\right] + \delta. \tag{2}
\end{align*}
To this end, fix a measurable set $W \subseteq \R$. We will need to use the following lemma later:

\begin{lemma}
Let $f_{n'}$ be the pdf of $\frac{V(\vec{t}_J)}{\sqrt{n'}} = \frac{1}{\sqrt{n'}} \sum_{j \in J} h(t_j)$ (where each $t_j \sim \T$ independently), and let $\phi$ be the pdf of the standard normal distribution. Then, $f_{n'}$ converges uniformly to $\phi$ as follows: 
\begin{align*}
\sup_{x \in \R} |f_{n'}(x) - \phi(x)| \leq \frac{c_1}{\sqrt{n'}},
\end{align*}
where $c_1 > 0$ is some (universal) constant.
\end{lemma}

\begin{proof}[Proof of lemma]
This lemma follows immediately from various ``local limit theorems'' (e.g., see Statement 5 in Section 4 of Chapter VII in \cite{Pet75}).
\end{proof}

Let $c_1$ be the constant in the above lemma, let $\gamma = \frac{e^\eps-1}{e^\eps+1}$, and let $\sigma^2 = \Var_{t' \sim \T}[g(t')]$. Let 
\begin{align*}
B = \left\{w \in \R : |w| > \sqrt{2\ln\left(\frac{\gamma \sqrt{n'}}{c_1\sqrt{2\pi}}\right)} - \frac{4\beta}{\sigma \sqrt{n'}}\right\}.
\end{align*}
Now, observe that
\begin{align*}
\Pr\left[\frac{V(t_i,\vec{t}_J)}{\sqrt{n'}} \in W\right] 
\leq \Pr\left[\frac{V(t_i,\vec{t}_J)}{\sqrt{n'}} \in W \cap \overline{B}\right] + \Pr\left[\frac{V(t_i,\vec{t}_J)}{\sqrt{n'}} \in B\right]. \tag{3}
\end{align*}

\begin{lemma}
\label{lem:twoChoicesCase1}
For each $w \in W \cap \overline{B}$,
\begin{align*}
\frac{\Pr[V(t_i, \vec{t}_J) / \sqrt{n'} = w]}{\Pr[V(t_i', \vec{t}_J) / \sqrt{n'} = w]} \leq \exp\left(\eps + O\left(\frac{\sqrt{\ln n'}}{\sqrt{n'}}\right)\right).
\end{align*}
\end{lemma}

\begin{proof}[Proof of lemma]
Fix $w \in W \cap \overline{B}$. Let $\phi$ be the pdf of the standard normal distribution. By the previous lemma, we have
\begin{align*}
\frac{\Pr[V(t_i, \vec{t}_J) / \sqrt{n'} = w]}{\Pr[V(t_i', \vec{t}_J) / \sqrt{n'} = w]}
= \ & \frac{\Pr[\frac{1}{\sqrt{n'}} \sum_{j \in J} h(t_j) = w - h(t_i) / \sqrt{n'}]}{\Pr[\frac{1}{\sqrt{n'}} \sum_{j \in J} h(t_j) = w - h(t_i') / \sqrt{n'}]} \\
\leq \ & \frac{\phi(w-h(t_i)/\sqrt{n'}) + \frac{c_1}{\sqrt{n'}}}{\phi(w-h(t_i')/\sqrt{n'}) - \frac{c_1}{\sqrt{n'}}}, \tag{4}
\end{align*}
where $c_1 \geq 0$ is some constant. Now, we note that
\begin{align*}
\frac{c_1}{\sqrt{n'}} \leq \gamma \cdot \phi(w - h(t_i)/\sqrt{n'}), \tag{5}
\end{align*}
since $\gamma \cdot \phi(w - h(t_i)/\sqrt{n'}) = \frac{\gamma}{\sqrt{2\pi}} \exp\left(-\frac{1}{2}\left(w - \frac{h(t_i)}{\sqrt{n'}}\right)^2\right) \geq \frac{\gamma}{\sqrt{2\pi}} \exp\left(-\frac{1}{2}\left(|w| + \frac{|h(t_i)|}{\sqrt{n'}}\right)^2\right)$, $|h(t_i)| \leq \frac{4\beta}{\sigma}$, and $|w| + \frac{4\beta}{\sigma \sqrt{n'}} \leq \sqrt{2\ln\left(\frac{\gamma \sqrt{n'}}{c_1\sqrt{2\pi}}\right)}$ (since $w \notin B$). Similarly, we also have
\begin{align*}
\frac{c_1}{\sqrt{n'}} \leq \gamma \cdot \phi(w - h(t_i')/\sqrt{n'}). \tag{6}
\end{align*}
Now, combining (4) with (5) and (6), we have 
\begin{align*}
\frac{\Pr[V(t_i, \vec{t}_J) / \sqrt{n'} = w]}{\Pr[V(t_i', \vec{t}_J) / \sqrt{n'} = w]}
\leq \ & \frac{\phi(w-h(t_i)/\sqrt{n'}) + \gamma \cdot \phi(w - h(t_i)/\sqrt{n'})}{\phi(w-h(t_i')/\sqrt{n'}) - \gamma \cdot \phi(w - h(t_i')/\sqrt{n'})} \\
= \ & \frac{1+\gamma}{1-\gamma} \cdot \frac{\phi(w-h(t_i)/\sqrt{n'})}{\phi(w-h(t_i')/\sqrt{n'})} \\
= \ & \exp(\eps) \cdot \exp\left(\frac{1}{2} \cdot \left(\left(w-\frac{h(t_i')}{\sqrt{n'}}\right)^2 - \left(w-\frac{h(t_i)}{\sqrt{n'}}\right)^2\right)\right) \\
\leq \ & \exp(\eps) \cdot \exp\left(\frac{1}{2} \cdot \left|\frac{h(t_i)}{\sqrt{n'}} - \frac{h(t_i')}{\sqrt{n'}}\right| \cdot \left|2w - \frac{h(t_i')}{\sqrt{n'}} - \frac{h(t_i)}{\sqrt{n'}}\right|\right) \\
\leq \ & \exp(\eps) \cdot \exp\left(\frac{1}{2} \cdot \frac{8\beta}{\sigma\sqrt{n'}} \cdot \left(2 \sqrt{2 \ln\left(\frac{\gamma\sqrt{n'}}{c_1\sqrt{2\pi}}\right)} \right)\right) \\
\leq \ & \exp\left(\eps + O\left(\frac{\sqrt{\ln n'}}{\sqrt{n'}}\right)\right),
\end{align*}
where the second last inequality follows from the fact that the function $h$ is bounded by $\frac{4\beta}{\sigma}$ (in absolute value) and $|w| + \frac{4\beta}{\sigma\sqrt{n'}} \leq \sqrt{2 \ln\left(\frac{\gamma\sqrt{n'}}{c_1\sqrt{2\pi}}\right)}$ (since $w \notin B$).
\end{proof}

We will need to use the following lemma later:

\begin{lemma}
Let $F_{n'}$ be the cdf of $\frac{V(\vec{t}_J)}{\sqrt{n'}} = \frac{1}{\sqrt{n'}} \sum_{j \in J} h(t_j)$ (where each $t_j \sim \T$ independently), and let $\Phi$ be the cdf of the standard normal distribution. Then, $\frac{V(\vec{t}_J)}{\sqrt{n'}}$ converges in distribution uniformly to the standard normal random variable as follows: 
\begin{align*}
\sup_{x \in \R} |F_{n'}(x) - \Phi(x)| = O\left(\frac{1}{\sqrt{n'}}\right).
\end{align*}
\end{lemma}

\begin{proof}[Proof of lemma]
This lemma follows immediately from the Berry-Esseen theorem.
\end{proof}

\begin{lemma}
\label{lem:twoChoicesCase2}
\begin{align*}
\Pr\left[\frac{V(t_i,\vec{t}_J)}{\sqrt{n'}} \in B\right] \leq O\left(\frac{1}{\eps \sqrt{n'}}\right).
\end{align*}
\end{lemma}

\begin{proof}[Proof of lemma]
Let $x = \sqrt{2\ln\left(\frac{\gamma \sqrt{n'}}{c_1\sqrt{2\pi}}\right)} - \frac{4\beta}{\sigma \sqrt{n'}}$, let $\Phi$ be the cdf of the standard normal distribution, and let $X \sim \mathcal{N}(0,1)$. Observe that
\begin{align*}
\Pr\left[\frac{V(t_i,\vec{t}_J)}{\sqrt{n'}} \in B\right] 
\leq \ & \Pr\left[\left|\frac{h(t_i)}{\sqrt{n'}}\right| + \left|\frac{V(\vec{t}_J)}{\sqrt{n'}}\right| > x \right] \\
\leq \ & 1 - \Pr\left[\left|\frac{V(\vec{t}_J)}{\sqrt{n'}}\right| \leq x - \frac{4\beta}{\sigma \sqrt{n'}} \right] \\
\leq \ & 1 - \left(\Pr\left[\frac{V(\vec{t}_J)}{\sqrt{n'}} \leq x - \frac{4\beta}{\sigma \sqrt{n'}} \right] - \Pr\left[\frac{V(\vec{t}_J)}{\sqrt{n'}} \leq - \left(x - \frac{4\beta}{\sigma \sqrt{n'}}\right)\right]\right) \\
\leq \ & 1 - \left(\Phi\left(x - \frac{4\beta}{\sigma \sqrt{n'}}\right) - \Phi\left(-\left(x - \frac{4\beta}{\sigma \sqrt{n'}}\right)\right) - O\left(\frac{1}{\sqrt{n'}}\right)\right) \\
= \ & 1 - \Pr\left[|X| \leq x - \frac{4\beta}{\sigma \sqrt{n'}}\right] + O\left(\frac{1}{\sqrt{n'}}\right) \\
= \ & \Pr\left[|X| > x - \frac{4\beta}{\sigma \sqrt{n'}}\right] + O\left(\frac{1}{\sqrt{n'}}\right),
\end{align*}
where the second inequality follows from the fact that $|h(t_i)| \leq \frac{4\beta}{\sigma}$, and the fourth inequality follows from the previous lemma. Now, by Mill's inequality, we have $\Pr[|X| \geq y] \leq \sqrt{\frac{2}{\pi}} \cdot \frac{1}{y} e^{-y^2/2}$ for every $y > 0$, so
\begin{align*}
\Pr\left[\frac{V(t_i,\vec{t}_J)}{\sqrt{n'}} \in B\right] 
\leq \ & \sqrt{\frac{2}{\pi}} \cdot \frac{1}{x - \frac{4\beta}{\sigma \sqrt{n'}}} \exp\left(-\frac{1}{2} \left(x - \frac{4\beta}{\sigma \sqrt{n'}}\right)^2 \right) + O\left(\frac{1}{\sqrt{n'}}\right) \\
= \ & O\left(\frac{1}{\gamma \sqrt{n'}}\right) \\
= \ & O\left(\frac{1}{\eps \sqrt{n'}}\right).
\end{align*}
\end{proof}

Now, by Lemmas \ref{lem:twoChoicesCase1} and \ref{lem:twoChoicesCase2}, we have from (3) that
\begin{align*}
\Pr\left[\frac{V(t_i,\vec{t}_J)}{\sqrt{n'}} \in W\right] 
\leq \ & \Pr\left[\frac{V(t_i,\vec{t}_J)}{\sqrt{n'}} \in W \cap \overline{B}\right] + \Pr\left[\frac{V(t_i,\vec{t}_J)}{\sqrt{n'}} \in B\right] \\
\leq \ & \exp\left(\eps + O\left(\frac{\sqrt{\ln n'}}{\sqrt{n'}}\right)\right) \cdot \Pr\left[\frac{V(t_i',\vec{t}_J)}{\sqrt{n'}} \in W \cap \overline{B}\right] + O\left(\frac{1}{\eps \sqrt{n'}}\right) \\
\leq \ & \exp\left(\eps + O\left(\frac{\sqrt{\ln (n-k)}}{\sqrt{n-k}}\right)\right) \cdot \Pr\left[\frac{V(t_i',\vec{t}_J)}{\sqrt{n'}} \in W\right] + O\left(\frac{1}{\eps \sqrt{n-k}}\right).
\end{align*}
This completes the proof.
\end{proof}

The mechanism in Theorem~\ref{thm:twoChoicesMechanism} chooses a social alternative by applying some function $f$ on the difference between $\sum_{i=1}^n u(t_i, A)$ and $\sum_{i=1}^n u(t_i, B)$.
We note that the mechanism may already have certain properties, such as efficiency, truthfulness, individual rationality, etc.; by Theorem~\ref{thm:twoChoicesMechanism}, this mechanism also satisfies privacy and persistent approximate truthfulness, in addition to the original properties that it already satisfies. One obvious application of Theorem~\ref{thm:twoChoicesMechanism} is to let $u$ be a common utility function for the players, where the utility of player $i$ with type $t_i$ is $u(t_i,A)$ if $A$ is chosen, and is $u(t_i,B)$ if $B$ is chosen. If we define $f: \R \to S$ such that $f(x)=A$ if and only if $x >0$, then the mechanism maximizes social welfare. 

\begin{example}[\bf{Public Project}]\label{ex:pp}
The canonical public project problem (see, e.g. \cite{NRTV07}) is that the government is trying to decide whether it should spend cost $c$ to build a certain public project (e.g., a bridge). 
We can apply Theorem~\ref{thm:twoChoicesMechanism} to this problem by letting social alternative $A$ represent building the project, and $B$ represent not building it. 
A player's type $t\in [-\alpha,\alpha]$ is the utility that he/she will gain from the public project being built.
Since not building the project results in a utility of $0$ to everyone, we let $u(t,B)=0$ for all $t$, and $u(t,A)=t$ for all $t$. 

In this case, if the distribution $\T$ has a density function and nonzero variance then so does $u(t,A)-u(t,B)$.
We can then define a {mechanism} that makes a choice based on the cost $c$ and the sum of the utilities $u(t,A)$ over the players.
By Theorem~\ref{thm:twoChoicesMechanism}, this mechanism is Bayesian differentially private and persistent approximately truthful. 
For example, if $n$ is the number of players, then from Theorem~\ref{thm:twoChoicesMechanism} we can conclude that the mechanism is $(\Omega(n),O(\frac{1}{\sqrt[4]{n}}), O(\frac{1}{\sqrt[4]{n}}))$-Bayesian differentially private and $(\Omega(n), \Omega(\sqrt[8]{n}))$-persistent $O(\frac{1}{\sqrt[8]{n}})$-Bayes-Nash truthful.
This means that even if a constant fraction of the players announce their types arbitrarily, and coalitions of size $O({\sqrt[8]{n}})$ can be formed, the privacy of each player is still protected, and any possible gain from lying is bounded by $O(\frac{1}{\sqrt[8]{n}})$.
As $n\rightarrow \infty$, any possible gain from lying converges to $0$, making the mechanism persistent Bayes-Nash truthful.

By building the project if and only if the sum of the utilities for $A$ exceeds the cost $c$, we have a mechanism that maximizes social welfare (here, we include the government's utility as part of the social welfare).
We also note that although the VCG mechanism can be used in this setting, it uses payments, which may not be appropriate. 
\end{example}

The functions $u$ in Theorem~\ref{thm:twoChoicesMechanism} can also be some arbitrary objective function, such as utilization/revenue, as in the next example.
%Example~\ref{ex:golf} also illustrates the fact that any objective function can be optimized when the mechanism is 
\begin{example}[\bf{Golf Course vs.~Low Cost Swimming Pool}]\label{ex:golf}
Suppose a recreational club wants to build either a golf course or a low-cost swimming pool. 
Based on past experience, a member's expected utilization of a golf course is a function $u(\cdot,A)$ that is increasing with respect to his income.
On the other hand, a member's expected utilization of a low-cost swimming pool is a function $u(\cdot,B)$ that is decreasing with respect to his income.
With this model, the club can use a mechanism that first asks each member $i$ to report his annual income $t_i$, then computes the difference $\sum_{i=1}^n u(t_i, A) - \sum_{i=1}^n u(t_i,B)$, and then makes a choice to maximize expected utilization (and hence revenue as well).
If the random variable $u(t,A) - u(t,B)$, where $t \sim \T$, has non-zero variance and a pdf, and the utility functions are bounded, then by Theorem~\ref{thm:twoChoicesMechanism}, this mechanism is Bayesian differentially private and persistent approximate Bayes-Nash truthful. 
\end{example}

\subsection{Social Welfare Mechanisms}
\label{sec:swMechanisms}

In this section, we present a class of mechanisms that make choices based on the social welfare provided by each social alternative. 
An important subset of these mechanisms is the set of mechanisms that maximize social welfare. 

In this section, a type $t\in T$ is a valuation function that assigns a valuation to each social alternative $s\in S$.
In many settings, it is reasonable to assume that the players' valuations for each social alternative follow a normal distribution, since the normal distribution has been frequently used to model many natural and social phenomena.
For convenience of presentation, we will use the \emph{standard} normal distribution $\mathcal{N}(0,1)$ in our theorems below. 
However, our theorems can be easily generalized to work with arbitrary normal distributions.
In any case, it is easy to see that given any normal distribution over the valuations, the valuations can be translated and scaled to obtain the standard normal distribution.

For any reasonable mechanism, it is natural to have a bound on the set of possible valuations --- it would be unreasonable to allow a player to report an arbitrarily high or low valuation (e.g. $2^{100}$) and single-handedly influence the choice of the mechanism. 
Therefore, we will restrict the possible valuations to the interval $[-\alpha,\alpha]$ for some value $\alpha>0$. 
As a result, our type space $T$ will be the set of all valuation functions $t: S \to [-\alpha,\alpha]$.
Furthermore, we will assume that the players' valuations for each social alternative follow the standard normal distribution.
However, because of the bound on the set of valuations, we will use the truncated standard normal distribution obtained by conditioning $\mathcal{N}(0,1)$ to lie on the interval $[-\alpha,\alpha]$. We denote this distribution by $\mathcal{N}(0,1)_{[-\alpha,\alpha]}$.

%For the results in this section, we assume that the common distribution $\T$ is the distribution obtained by letting $t(s)\sim \mathcal{N}(0,1)_{[-\alpha,\alpha]}$ for each $s\in S$ independently. 
For simplicity, we will first present the following theorem, which is a special case of our more general result (Theorem~\ref{thm:generalSWMechanism}).
The following theorem states that if each player's valuation for each social alternative is distributed as the truncated standard normal distribution $\mathcal{N}(0,1)_{[-\alpha,\alpha]}$, then any mechanism that makes a choice based on the set of total valuations for each social alternative is Bayesian differentially private and persistent approximate Bayes-Nash truthful.

%theorem goes here
\begin{theorem}[{\bf Social welfare mechanisms}]
\label{thm:sw}
Let $S = \{s_1, \ldots, s_m\}$ be a set of $m$ social alternatives. Let the type space $T$ be the set of all valuation functions $t: S \to [-\alpha,\alpha]$ on $S$, where $\alpha = \Theta(\sqrt[4]{n})$. Let $\T$ be the distribution over $T$ obtained by letting $t(s)\sim \mathcal{N}(0,1)_{[-\alpha,\alpha]}$ for each $s \in S$ independently. For each player $i$, let the utility function for player $i$ be $u_i(t_i,s) = t_i(s)$.

Let $\mathsf{sw}_j(\vec{t}) = \sum_{i=1}^n t_i(s_j)$ be the (reported) social welfare for the social alternative $s_j$. Let $M: T^n \to S$ be any mechanism such that 
\begin{align*}
M(\vec{t}) = f(\mathsf{sw}_1(\vec{t}), \ldots, \mathsf{sw}_m(\vec{t}))
\end{align*}
for some function $f: \R^m \to S$. Then, for every constant $c < 1$, every $k \leq c \cdot n$, and every $0 < \eps \leq 1$, $M$ satisfies the following properties with $\delta = O(e^{-\Omega(\frac{\eps^2}{m^2} \cdot \sqrt{n}) + \ln (m\sqrt{n})})$:
\begin{enumerate}
  \item \emph{Privacy:} $M$ is $(k,\eps,\delta)$-Bayesian differentially private.
  \item \emph{Persistent approximate truthfulness:} For every $1 \leq r \leq k+1$, $M$ is $(k-r+1,r)$-persistent $(r\eps + 2r\delta)(2\alpha)$-Bayes-Nash truthful.
\end{enumerate}
\end{theorem}

In Theorem~\ref{thm:sw}, $\mathsf{sw}_j(\vec{t})$ represents the social welfare that will be achieved if the players' types (i.e., valuation functions) are $\vec{t}$ and the social alternative $s_j$ is chosen by the mechanism.
Thus, Theorem~\ref{thm:sw} says that any mechanism whose choice depends only on the set $\{ \mathsf{sw}_j(\vec{t})\}_{j\in [m]}$ of social welfare values satisfies Bayesian differential privacy and persistent approximate Bayes-Nash truthfulness, in addition to any properties that it may already satisfy (e.g., efficiency, truthfulness, individual rationality, etc.).
In particular, a mechanism that chooses a social alternative to maximize social welfare satisfies this requirement and achieves optimal efficiency with respect to social welfare. 

In Theorem~\ref{thm:sw}, the value $\alpha$ at which the standard normal distribution is truncated is chosen so that the truncated distribution is very close to the untruncated one.
This means that even if we had used the untruncated distribution instead, with high probability no valuation would fall outside the interval $[-\alpha,\alpha]$.

In the next theorem, we consider a setting where there is a set of available ``options'', and we allow the mechanism to choose any subset of these options. 
Thus, the set of social alternatives is the power set of the set of options. 
To keep the set of valuations tractable, instead of having a valuation for each social alternative, the players have a valuation for each option. 
Moreover, we allow for the flexibility where for each player, only certain options are relevant/applicable to that player.
We capture this flexibility by having a binary weight for each player-option pair. 
Note that Theorem~\ref{thm:sw} is the special case where the set of social alternatives consists of the sets of single options (i.e., the singletons), and where all options are considered relevant to all players.

The binary weight $w_{i,j}$ associated with player $i$ and option $o_j$ indicates whether option $o_j$ is relevant/applicable to player $i$. 
$w_{i,j}=1$ means that option $o_j$ is relevant/applicable to player $i$, so player $i$'s announced valuation is taken into account in the social welfare for option $o_j$.
On the other hand, $w_{i,j}=0$ means that player $i$'s valuation is ignored in the social welfare for option $o_j$. 
These weights are known to or set by the mechanism designer.
For example, perhaps only people with low income should have a voice in decisions regarding subsidized housing, and only people with disabilities should have a say in decisions regarding building accessibility laws. We now state our next theorem, which generalizes Theorem~\ref{thm:sw} to this new setting.

%generalized theorem goes here 
\begin{theorem}[{\bf Social welfare mechanisms with multiple options}]
\label{thm:generalSWMechanism}
Let the set $S$ of social alternatives be $2^O$, where $O = \{o_1, \ldots, o_m\}$ is a set of $m$ possible ``options''. Let the type space $T$ be the set of all valuation functions $t: O \to [-\alpha,\alpha]$ on $O$, where $\alpha = \Theta(\sqrt[4]{n})$. Let $\T$ be the distribution over $T$ obtained by letting $t(o)\sim\mathcal{N}(0,1)_{[-\alpha,\alpha]}$ for each option $o \in O$ independently. Suppose the weights $\{w_{i,j}\}_{i\in [n],j\in [m]}$ satisfy $\sum_{i=1}^n w_{i,j} \geq c_1 \cdot n$ for every option $o_j$, where $c_1 > 0$ is some constant. 

Let $\mathsf{sw}_j(\vec{t}) = \sum_{i=1}^n w_{i,j} \cdot t_i(o_j)$ be the (reported) social welfare for option $o_j$. Let $M: T^n \to S$ be any mechanism such that 
\begin{align*}
M(\vec{t}) = f(\mathsf{sw}_1(\vec{t}), \ldots, \mathsf{sw}_m(\vec{t}))
\end{align*}
for some function $f: \R^m \to S$. Then, for every constant $c_2 < c_1$, every $k \leq c_2 \cdot n$, and every $0 < \eps \leq 1$, $M$ satisfies the following properties with $\delta = O(e^{-\Omega(\frac{\eps^2}{m^2} \cdot \sqrt{n}) + \ln (m\sqrt{n})})$:
\begin{enumerate}
  \item \emph{Privacy:} $M$ is $(k,\eps,\delta)$-Bayesian differentially private.
  \item \emph{Persistent approximate truthfulness:} Suppose the utility functions are bounded by $\beta > 0$, i.e., the utility function for each player $i$ is $u_i: T \times S \to [-\beta,\beta]$. Then, for every $1 \leq r \leq k+1$, $M$ is $(k-r+1,r)$-persistent $(r\eps + 2r\delta)(2\beta)$-Bayes-Nash truthful.
\end{enumerate}
\end{theorem}

The proof idea \emph{roughly} works as follows. Persistent approximate truthfulness follows from Bayesian differential privacy and Theorem \ref{thm:BDP=>Truthfulness}, so we only have to show that the mechanism $M$ is Bayesian differentially private. 

Since the mechanism chooses a social alternative based only on $(\sw_1(\vec{t}), \ldots, \sw_m(\vec{t}))$, it suffices to show that $(\sw_1(\vec{t}), \ldots, \sw_m(\vec{t}))$ is Bayesian differentially private. Now, consider a player $i$ changing his/her announced type from $t_i$ to $t_i'$. This changes only one of the components of $(\sw_1(\vec{t}), \ldots, \sw_m(\vec{t}))$ by at most $2\alpha$. However, the randomness of the other players' types can easily blur/smooth out this change. Let $\vec{Z}$ be the random variable $(\sw_1(\vec{t}_{-i}), \ldots, \sw_m(\vec{t}_{-i}))$ where the other players' types $\vec{t}_{-i}$ are each distributed according to $\T$ in the theorem, which uses the truncated standard normal distribution $\mathcal{N}(0,1)_{[-\alpha,\alpha]}$. Let $\vec{Z}'$ be the same as $\vec{Z}$ except that instead of using $\mathcal{N}(0,1)_{[-\alpha,\alpha]}$, $\vec{Z}'$ uses the actual standard normal distribution $\mathcal{N}(0,1)$. 

We first note that the distribution of $\vec{Z}$ is a good approximation of the distribution of $\vec{Z}'$. Then, we note that $\vec{Z}$ is a random vector with independent components, each of which has a normal distribution with $\Omega(n)$ variance. Each of these normal distributions is relatively smooth at every point that is not too far out into the tails, i.e., the probability density at two nearby points is relatively close to each other, provided that the two points are not too far out into the tails. Now, we observe that a sample from these normal distributions will not be too far out into the tails with high probability. Together, these facts imply that the distribution of $(\sw_1(\vec{t}), \ldots, \sw_m(\vec{t}))$ when player $i$ reports $t_i$ is relatively close to the distribution when player $i$ reports $t_i'$. This shows that $(\sw_1(\vec{t}), \ldots, \sw_m(\vec{t}))$ is Bayesian differentially private.

\begin{proof}
The second property follows from the first property and Theorem \ref{thm:BDP=>Truthfulness}. We now show the first property. Let $i \in [n]$, let $I \subseteq [n] \setminus \{i\}$ with $|I| \leq k$, let $t_i, t_i' \in T$, let $\vec{t}'_I \in T^{|I|}$, and let $Y \subseteq S$. Let $J = [n] \setminus (I \cup \{i\})$ and $\vec{t}_J \sim \T^{|J|}$. Let $n' = |J|$. Given a subset $A$ of players with reported types $\vec{t}_A$, let $\mathsf{sw}_\ell(\vec{t}_A) = \sum_{j \in A} w_{j,\ell} \cdot t_j(o_\ell)$. We need to show that
\begin{align*}
  \Pr[M(t_i, \vec{t}'_I, \vec{t}_J) \in Y] \leq e^{\eps} \Pr[M(t_i', \vec{t}'_I, \vec{t}_J) \in Y] + \delta,
\end{align*}
which is equivalent to 
\begin{align*}
\Pr[f(\mathsf{sw}_1(t_i, \vec{t}'_I, \vec{t}_J), \ldots, \mathsf{sw}_m(t_i, \vec{t}'_I, \vec{t}_J)) \in Y] \leq e^{\eps} \Pr[f(\mathsf{sw}_1(t_i', \vec{t}'_I, \vec{t}_J), \ldots, \mathsf{sw}_m(t_i', \vec{t}'_I, \vec{t}_J)) \in Y] + \delta \tag{1}
\end{align*}
for some function $f: \R \to S$. To show (1), it is easy to see that it suffices to show that for every measurable set $W \subseteq \R^m$,
\begin{align*}
\Pr[(\mathsf{sw}_1(t_i,\vec{t}_J), \ldots, \mathsf{sw}_m(t_i, \vec{t}_J)) \in W] \leq e^{\eps} \Pr[(\mathsf{sw}_1(t_i',\vec{t}_J), \ldots, \mathsf{sw}_m(t_i', \vec{t}_J)) \in W] + \delta. \tag{2}
\end{align*}
To this end, fix a measurable set $W \subseteq \R^m$. Let $\sigma_{\min}^2 = \displaystyle\min_{\ell \in [m]} \sum_{j \in J} w_{j,\ell}$, and let 
\begin{align*}
B = \{\vec{w} \in \R^m : |w_\ell| > \frac{\eps \sigma_{\min}^2}{2m\alpha} - \alpha \text{ for some } \ell \in [m]\}. 
\end{align*}
Now, observe that
\begin{align*}
& \Pr[(\mathsf{sw}_1(t_i,\vec{t}_J), \ldots, \mathsf{sw}_m(t_i, \vec{t}_J)) \in W] \\
\leq \ & \Pr[(\mathsf{sw}_1(t_i,\vec{t}_J), \ldots, \mathsf{sw}_m(t_i, \vec{t}_J)) \in W \cap \overline{B}] + \Pr[(\mathsf{sw}_1(t_i,\vec{t}_J), \ldots, \mathsf{sw}_m(t_i, \vec{t}_J)) \in B]. \tag{3}
\end{align*}

\begin{lemma}
\label{lem:swMechCase1}
\begin{align*}
\ & \Pr[(\mathsf{sw}_1(t_i,\vec{t}_J), \ldots, \mathsf{sw}_m(t_i, \vec{t}_J)) \in W \cap \overline{B}] \\
\leq \ & e^{\eps} \Pr[(\mathsf{sw}_1(t_i',\vec{t}_J), \ldots, \mathsf{sw}_m(t_i', \vec{t}_J)) \in W \cap \overline{B}] + m e^{-\Omega(\sqrt{n})}.
\end{align*}
\end{lemma}

\begin{proof}[Proof of lemma]
Let $\vec{Z} = (\sw_1(\vec{t}_J), \ldots, \sw_m(\vec{t}_J))$. Let $\vec{sw}(t_i) = (\sw_1(t_i), \ldots, \sw_m(t_i))$ and $\vec{sw}(t_i') = (\sw_1(t_i'), \ldots, \sw_m(t_i'))$. Let $W \cap \overline{B} - \vec{sw}(t_i) = \{\vec{w} - \vec{sw}(t_i) : \vec{w} \in W \cap \overline{B}\}$ and $W \cap \overline{B} - \vec{sw}(t_i') = \{\vec{w} - \vec{sw}(t_i') : \vec{w} \in W \cap \overline{B}\}$. Then, it suffices to show that
\begin{align*}
\Pr[\vec{Z} \in W \cap \overline{B} - \vec{sw}(t_i)] \leq e^{\eps} \Pr[\vec{Z} \in W \cap \overline{B} - \vec{sw}(t_i')] + m e^{-\Omega(\sqrt{n})}. \tag{4}
\end{align*}

Let $\vec{Z}' = (\sw_1(\vec{t}_J'), \ldots, \sw_m(\vec{t}_J'))$, where for each $j \in J$, we have $t_j'(o_\ell) \sim \mathcal{N}(0,1)$ for every $\ell \in [m]$ independently. 

Recall that for each $j \in J$ and $\ell \in [m]$, we have $t_j(o_\ell) \sim \mathcal{N}(0,1)_{[-\alpha,\alpha]}$. Let $f_{t_j(o_\ell)}$ be the pdf of $t_j(o_\ell)$, and let $\phi$ be the pdf of the standard normal distribution. Now, for $X \sim \mathcal{N}(0,1)$, by Mill's inequality, we have $\Pr[|X| \geq \alpha] \leq \sqrt{\frac{2}{\pi}} (\frac{1}{\alpha}) e^{-\Omega(\alpha^2)} = O(\frac{1}{n^{1/4}} e^{-\Omega(\sqrt{n})})$. Then, it is easy to see that $||f_{t_j(o_\ell)} - \phi||_1 = \int_\R |f_{t_j(o_\ell)}(x) - \phi(x)| \ dx \leq O(\frac{1}{n^{1/4}} e^{-\Omega(\sqrt{n})})$ for each $j \in J$ and $\ell \in [m]$. It follows that $||pdf((t_j(o_\ell))_{j \in J, \ell \in [m]}) - pdf((t_j'(o_\ell))_{j \in J, \ell \in [m]})||_1 \leq O(n'm \cdot \frac{1}{n^{1/4}} e^{-\Omega(\sqrt{n})}) = m e^{-\Omega(\sqrt{n})}$, where $pdf(\vec{X})$ denotes the pdf of the random vector $\vec{X}$. Thus, for every measurable set $W' \subseteq \R^m$, we have
\begin{align*}
|\Pr[\vec{Z} \in W'] - \Pr[\vec{Z}' \in W']| \leq m e^{-\Omega(\sqrt{n})}. \tag{5}
\end{align*} 
Let $f_{\vec{Z}'}$ be the pdf of $\vec{Z}'$. Then, to show (4), it suffices to show that for every $\vec{v} \in W \cap \overline{B}$,
\begin{align*}
\frac{f_{\vec{Z}'}(\vec{v} - \vec{sw}(t_i))}{f_{\vec{Z}'}(\vec{v} - \vec{sw}(t_i'))} \leq e^{\eps}, \tag{6}
\end{align*}
since then we would have
\begin{align*}
\Pr[\vec{Z} \in W \cap \overline{B} - \vec{sw}(t_i)] 
\leq \ & \Pr[\vec{Z}' \in W \cap \overline{B} - \vec{sw}(t_i)] + m e^{-\Omega(\sqrt{n})} \\
\leq \ & e^{\eps} \Pr[\vec{Z}' \in W \cap \overline{B} - \vec{sw}(t_i')] + m e^{-\Omega(\sqrt{n})} \\
\leq \ & e^{\eps} (\Pr[\vec{Z} \in W \cap \overline{B} - \vec{sw}(t_i')] + m e^{-\Omega(\sqrt{n})}) + m e^{-\Omega(\sqrt{n})} \\
\leq \ & e^{\eps} \Pr[\vec{Z} \in W \cap \overline{B} - \vec{sw}(t_i')] + m e^{-\Omega(\sqrt{n})},
\end{align*}
where the first and third inequality follows from (5), and the second inequality follows from (6). Thus, we now show (6). 

Fix $\vec{v} = (v_1, \ldots, v_m) \in W \cap \overline{B}$. Recall that $\vec{Z}' = (\sw_1(\vec{t}_J'), \ldots, \sw_m(\vec{t}_J'))$. Thus, we have $\vec{Z}' \sim (\mathcal{N}(0,\sigma_1^2), \ldots, \mathcal{N}(0,\sigma_m^2))$, where $\sigma_{\ell}^2 = \sum_{j \in J} w_{j,\ell}$ for $\ell \in [m]$. Now, observe that
\begin{align*}
\frac{f_{\vec{Z'}}(\vec{v}-\vec{sw}(t_i))}{f_{\vec{Z'}}(\vec{v}-\vec{sw}(t_i'))}
= \ & \prod_{\ell = 1}^m \exp\left(\frac{1}{2 \sigma_\ell^2} \cdot ((v_\ell - \sw_\ell(t_i'))^2 - (v_\ell - \sw_\ell(t_i))^2) \right) \\
\leq \ & \prod_{\ell = 1}^m \exp\left(\frac{1}{2 \sigma_\ell^2} \cdot |(\sw_\ell(t_i) - \sw_\ell(t_i'))| \cdot |2v_\ell - \sw_\ell(t_i) - \sw_\ell(t_i')| \right) \\
\leq \ & \prod_{\ell = 1}^m \exp\left(\frac{2\alpha}{\sigma_\ell^2} \cdot (|v_\ell| + \alpha) \right) \\
\leq \ & \prod_{\ell = 1}^m \exp\left(\frac{\eps \sigma_{\min}^2}{m \sigma_\ell^2} \right) \\
\leq \ & \exp(\eps),
\end{align*} 
where the third inequality follows from the fact that $\vec{v} \notin B$. This completes the proof of the lemma. 
\end{proof}

\begin{lemma}
\label{lem:swMechCase2}
\begin{align*}
\Pr[(\mathsf{sw}_1(t_i,\vec{t}_J), \ldots, \mathsf{sw}_m(t_i, \vec{t}_J)) \in B] \leq O(e^{-\Omega(\frac{\eps^2}{m^2} \cdot \sqrt{n}) + \ln (m\sqrt{n})}).
\end{align*}
\end{lemma}

\begin{proof}[Proof of lemma]
By the union bound, we have
\begin{align*}
\Pr[(\mathsf{sw}_1(t_i,\vec{t}_J), \ldots, \mathsf{sw}_m(t_i, \vec{t}_J)) \in B]
\leq \ & \sum_{\ell = 1}^m \Pr[|\mathsf{sw}_\ell(t_i,\vec{t}_J)| > \frac{\eps\sigma_{\min}^2}{2\alpha m} - \alpha]. \tag{7}
\end{align*}
Now, fix $\ell \in [m]$, and observe that
\begin{align*}
\Pr[|\mathsf{sw}_\ell(t_i,\vec{t}_J)| > \frac{\eps\sigma_{\min}^2}{2\alpha m} - \alpha]
= \ & \Pr[|\sw_\ell(t_i) + \mathsf{sw}_\ell(\vec{t}_J)| > \frac{\eps\sigma_{\min}^2}{2\alpha m} - \alpha] \\
\leq \ & \Pr[|\mathsf{sw}_\ell(\vec{t}_J)| > \frac{\eps\sigma_{\min}^2}{2\alpha m} - 2\alpha].
\end{align*}

Let $\sw_{\ell}(\vec{t}_J') = \sum_{j \in J} \sw_{\ell}(t_j')$, where for each $j \in J$, we have $t_j'(o_{\ell}) \sim \mathcal{N}(0,1)$ independently. 

Recall that for each $j \in J$, we have $t_j(o_\ell) \sim \mathcal{N}(0,1)_{[-\alpha,\alpha]}$. Let $f_{t_j(o_\ell)}$ be the pdf of $t_j(o_\ell)$, and let $\phi$ be the pdf of the standard normal distribution. Now, for $X \sim \mathcal{N}(0,1)$, by Mill's inequality, we have $\Pr[|X| \geq \alpha] \leq \sqrt{\frac{2}{\pi}} (\frac{1}{\alpha}) e^{-\Omega(\alpha^2)} = O(\frac{1}{n^{1/4}} e^{-\Omega(\sqrt{n})})$. Then, it is easy to see that $||f_{t_j(o_\ell)} - \phi||_1 = \int_\R |f_{t_j(o_\ell)}(x) - \phi(x)| \ dx \leq O(\frac{1}{n^{1/4}} e^{-\Omega(\sqrt{n})})$ for each $j \in J$. It follows that $||pdf((t_j(o_\ell))_{j \in J}) - pdf((t_j'(o_\ell))_{j \in J})||_1 \leq O(n' \cdot \frac{1}{n^{1/4}} e^{-\Omega(\sqrt{n})}) = O(e^{-\Omega(\sqrt{n})})$, where $pdf(\vec{X})$ denotes the pdf of the random vector $\vec{X}$. Thus, we have
\begin{align*}
\Pr\left[|\mathsf{sw}_\ell(\vec{t}_J)| > \frac{\eps\sigma_{\min}^2}{2\alpha m} - 2\alpha\right] \leq \Pr\left[|\mathsf{sw}_\ell(\vec{t}_J')| > \frac{\eps\sigma_{\min}^2}{2\alpha m} - 2\alpha\right] + O(e^{-\Omega(\sqrt{n})}).
\end{align*}

Now, observe that $\mathsf{sw}_\ell(\vec{t}_J') \sim \mathcal{N}(0,\sigma_{\ell}^2)$, where $\sigma_{\ell}^2 = \sum_{j \in J} w_{j,\ell} \leq n$. Then, by Mill's inequality, we have $\Pr[|\mathsf{sw}_\ell(\vec{t}_J')| \geq x] \leq \sqrt{\frac{2}{\pi}} \cdot \frac{\sigma_\ell}{x} e^{-x^2/(2\sigma_{\ell}^2)} \leq O(\frac{\sqrt{n}}{x} e^{-x^2/(2n)})$ for $x > 0$, so
\begin{align*}
\Pr\left[|\mathsf{sw}_\ell(\vec{t}_J')| > \frac{\eps\sigma_{\min}^2}{2\alpha m} - 2\alpha\right] 
\leq \ & O\left(\frac{\sqrt{n}}{\frac{\eps\sigma_{\min}^2}{2\alpha m} - 2\alpha} \cdot \exp\left(-\frac{\left(\frac{\eps\sigma_{\min}^2}{2\alpha m} - 2\alpha\right)^2}{2n}\right)\right) \\
\leq \ & O\left(\exp\left(-\Omega\left(\frac{\eps^2}{m^2} \cdot \sqrt{n}\right) + \ln \sqrt{n} \right)\right),
\end{align*}
where the last inequality follows from the fact that $\alpha = \Theta(n^{1/4})$ (from the theorem statement) and $\sigma_{\min}^2 = \Omega(n)$ (by definition of $\sigma_{\min}^2$ and the fact that $\sum_{j \in J} w_{j,\ell} \geq (c_1 - c_2) \cdot n - 1$ for every $\ell \in [m]$). Now, combining the above results, we get
\begin{align*}
\Pr[(\mathsf{sw}_1(t_i,\vec{t}_J), \ldots, \mathsf{sw}_m(t_i, \vec{t}_J)) \in B] 
\leq \ & O(m \cdot e^{-\Omega(\frac{\eps^2}{m^2} \cdot \sqrt{n}) + \ln \sqrt{n}}) + O(m e^{-\Omega(\sqrt{n})}) \\
\leq \ & O(e^{-\Omega(\frac{\eps^2}{m^2} \cdot \sqrt{n}) + \ln (m\sqrt{n})}).
\end{align*}
This completes the proof of the lemma.
\end{proof}

Now, by Lemmas \ref{lem:swMechCase1} and \ref{lem:swMechCase2}, we have from (3) that
\begin{align*}
& \Pr[(\mathsf{sw}_1(t_i,\vec{t}_J), \ldots, \mathsf{sw}_m(t_i, \vec{t}_J)) \in W] \\
\leq \ & \Pr[(\mathsf{sw}_1(t_i,\vec{t}_J), \ldots, \mathsf{sw}_m(t_i, \vec{t}_J)) \in W \cap \overline{B}] + \Pr[(\mathsf{sw}_1(t_i,\vec{t}_J), \ldots, \mathsf{sw}_m(t_i, \vec{t}_J)) \in B] \\
\leq \ & e^{\eps} \Pr[(\mathsf{sw}_1(t_i',\vec{t}_J), \ldots, \mathsf{sw}_m(t_i', \vec{t}_J)) \in W \cap \overline{B}] + m e^{-\Omega(\sqrt{n})} + O(e^{-\Omega(\frac{\eps^2}{m^2} \cdot \sqrt{n}) + \ln (m\sqrt{n})}) \\
\leq \ & e^{\eps} \Pr[(\mathsf{sw}_1(t_i',\vec{t}_J), \ldots, \mathsf{sw}_m(t_i', \vec{t}_J)) \in W] + O(e^{-\Omega(\frac{\eps^2}{m^2} \cdot \sqrt{n}) + \ln (m\sqrt{n})}).
\end{align*}
This completes the proof.
\end{proof}

In Theorem~\ref{thm:generalSWMechanism}, the requirement on the binary weights simply means that each option is relevant/applicable to at least some constant fraction of the players. Note that the persistent approximate truthfulness result of Theorem~\ref{thm:generalSWMechanism} requires the players' utility functions to be bounded by $\beta>0$. 
This assumption is needed since the players' utility functions can actually be arbitrary functions. 
However, the most natural way to use Theorem~\ref{thm:generalSWMechanism} is to let player $i$'s utility function be the following:
 if the chosen social alternative is a singleton $\{o_j\}$, then the utility for player $i$ is $w_{i,j} \cdot t_i(o_j)$;
 if the chosen social alternative is a set $s$ consisting of two or more options, then the utility for player $i$ is the sum of the utilities for each singleton subset of $s$.
Alternatively, a player $i$'s utility for a social alternative $s$ does not have to be \emph{additive} in the options that $s$ contains --- the utility function for player $i$ can capture complementarities and substitutabilities of the options as well. We now give a simple example that illustrates Theorem~\ref{thm:generalSWMechanism}.

\begin{example}[\bf{Multiple Public Projects}]\label{ex:mpp}
The municipal government would like to spend its budget surplus of $4$ million on the community.
There are four options that the government is considering, each costing $2$ million to build: a senior home; a casino; a subsidized housing complex; or a library. 
The government would like to find out, on a scale from $-\alpha$ to $\alpha$, how much each individual values each option.
For each individual $i$, the government chooses the weights for each of the options as follows:
the weight for the senior home is $1$ if and only if individual $i$ is over the age of $65$; the weight for the casino is $1$ if and only if individual $i$ is over the age of $19$; the weight for the subsidized housing complex is $1$ if and only if individual $i$ is classified as low-income; and the weight for the library is always $1$.

After collecting the valuation functions from the individuals, the government can compute the social welfare provided by each option, or compute an average utility for each option by dividing its social welfare by the number of people who have weight $1$ for that option. 
Finally, the government can choose two of the options to maximize social welfare or average utility.

By Theorem~\ref{thm:generalSWMechanism}, this mechanism is Bayesian differentially private and persistent approximately truthful. For example, if $n$ is the number of individuals, then from Theorem~\ref{thm:generalSWMechanism} we can conclude that the mechanism is $(\Omega(n), n^{-1/5},e^{-\Omega(n^{1/10})})$-Bayesian differentially private and $(\Omega(n),\Omega(n^{1/10}))$-persistent $O(\frac{1}{n^{1/10}})$-Bayes-Nash truthful.
This means that even if a constant fraction of the individuals report their valuations arbitrarily, and coalitions of size $O(n^{1/10})$ can be formed, the privacy of each individual is still protected, and any possible gain from lying is bounded by $O(\frac{1}{n^{1/10}})$.
As $n\rightarrow \infty$, any possible gain from lying converges to $0$, making the mechanism persistent Bayes-Nash truthful.
\end{example}

\begin{example}[\bf{Group Auctions}]\label{ex:group}
Suppose that $m$ organizations are bidding to host the next WINE conference, and the conference organizers want to give the hosting privilege to the organization whose employees value the privilege the most. 
The organizers decide to poll the employees of each organization to see, on a scale from $-\alpha$ to $\alpha$, how much they value their organization hosting the event.
There are $m$ options, one corresponding to each organization winning the bid.
Each employee's type is a valuation function $t: O \to [-\alpha,\alpha]$, where $t(o_j)$ represents the utility he receives if organization $j$ wins the bid. 
Naturally, the employees of each organization should only be able to specify their valuations for their own organization winning; thus, the weight $w_{i,j}$ for player $i$ and organization $j$ is $1$ if $i$ works at $j$, and is $0$ otherwise. 

The organizers can compute the social welfare provided by each option, and then make a choice to maximize social welfare. Alternatively, to be fair to smaller organizations, the organizers can maximize the average utility in the winning organization.
We can use Theorem~\ref{thm:generalSWMechanism} to conclude that such a mechanism is Bayesian differentially private and persistent approximate Bayes-Nash truthful.
\end{example}

%This example has some interesting variations. 
%First, a special case of weighted social welfare is one where everyone has weight $1$ for every option. 
%For example, if the options were tax cut, carbon footprint reduction, and cancer research, naturally everyone's voice should be included. 
%Moreover, it is not unreasonable to model an individual's view of tax cut as being independent of his view of cancer research and carbon footprint reduction.
%
%Secondly, even if the utility of an individual is assumed to be additive in each option, the costs of combinations of options need not be additive. 
%For example, if the government decides to build both a senior home and subsidized housing, the cost may be significantly less than the sum of the cost of constructing each option alone, adding a combinatoric flavor to the problem. 
%
%Lastly, the player utilities can display some complementarities and substitutabilities as well.
%For example, a library and a casino both occupy people's free time, so people might feel less strongly about the pair together than the sum of the two. 
%In this case, the social welfare of $\{library,casino\}$ might be $\frac{3}{4} (\sum_{i=1}^n w_{i,library} t_i(library) + \sum_{i=1}^n w_{i,casino} t_i(casino))$. 
%On the other hand, senior home residents have lots of free time to spend at the library, so people's feelings might be amplified when the two are built together. 
%In this case, the social welfare provided by $\{library,senior\}$ might be $\frac{4}{3} (\sum_{i=1}^n w_{i,library} t_i(library) + \sum_{i=1}^n w_{i,senior} t_i(senior))$. 

We now generalize Theorem \ref{thm:generalSWMechanism} to the case where the players' valuations for each social alternative are arbitrarily distributed with non-zero variance and a pdf, although the parameters are worse. For convenience, we will restrict the valuations to be in a constant size interval $[-\alpha,\alpha]$, where $\alpha > 0$ is a constant instead of $\alpha = \Theta(\sqrt[4]{n})$ like in Theorem \ref{thm:generalSWMechanism}. (This restriction improves the parameters slightly.)

\begin{theorem}[{\bf Social welfare mechanisms with multiple options and more general distributions}]
\label{thm:generalSWMechanism2}
Let the set $S$ of social alternatives be $2^O$, where $O = \{o_1, \ldots, o_m\}$ is a set of $m$ possible ``options''. Let the type space $T$ be the set of all valuation functions $t: O \to [-\alpha,\alpha]$ on $O$, where $\alpha > 0$ is a constant. For each option $o \in O$, let $\T_o$ be any distribution over $[-\alpha,\alpha]$ with non-zero variance and a probability density function. Let $\T$ be the distribution over $T$ obtained by letting $t(o) \sim \T_o$ for each $o \in O$ independently. Suppose the weights $\{w_{i,j}\}_{i\in [n],j\in [m]}$ satisfy $\sum_{i=1}^n w_{i,j} \geq c_1 \cdot n$ for every option $o_j$, where $c_1 > 0$ is some constant. 

Let $\mathsf{sw}_j(\vec{t}) = \sum_{i=1}^n w_{i,j} \cdot t_i(o_j)$ be the (reported) social welfare for option $o_j$. Let $M: T^n \to S$ be any mechanism such that 
\begin{align*}
M(\vec{t}) = f(\mathsf{sw}_1(\vec{t}), \ldots, \mathsf{sw}_m(\vec{t}))
\end{align*}
for some function $f: \R^m \to S$. Then, for every constant $c_2 < c_1$, every $k \leq c_2 \cdot n$, and every $0 < \eps \leq 1$, $M$ satisfies the following properties with $\eps' = \eps + O(\frac{m \sqrt{\ln n}}{\sqrt{n}} )$ and $\delta = O(\frac{m^2}{\eps \sqrt{n}})$:
\begin{enumerate}
  \item \emph{Privacy:} $M$ is $(k,\eps',\delta)$-Bayesian differentially private.
  \item \emph{Persistent approximate truthfulness:} Suppose the utility functions are bounded by $\beta > 0$, i.e., the utility function for each player $i$ is $u_i: T \times S \to [-\beta,\beta]$. Then, for every $1 \leq r \leq k+1$, $M$ is $(k-r+1,r)$-persistent $(r\eps' + 2r\delta)(2\beta)$-Bayes-Nash truthful.
\end{enumerate}
\end{theorem}

The proof is similar to the proof of Theorem \ref{thm:twoChoicesMechanism}, except that we now have to deal with a more general setting (e.g., there are more than two social alternatives, there are options and weights, etc.).

\begin{proof} 
The second property follows from the first property and Theorem \ref{thm:BDP=>Truthfulness}. We now show the first property.

Let $i \in [n]$, let $I \subseteq [n] \setminus \{i\}$ with $|I| \leq k$, let $t_i, t_i' \in T$, let $\vec{t}'_I \in T^{|I|}$, and let $Y \subseteq S$. Let $J = [n] \setminus (I \cup \{i\})$ and $\vec{t}_J \sim \T^{|J|}$. Given a subset $A$ of players with reported types $\vec{t}_A$, let $\mathsf{sw}_\ell(\vec{t}_A) = \sum_{j \in A} w_{j,\ell} \cdot t_j(o_\ell)$. We need to show that
\begin{align*}
  \Pr[M(t_i, \vec{t}'_I, \vec{t}_J) \in Y] \leq e^{\eps'} \Pr[M(t_i', \vec{t}'_I, \vec{t}_J) \in Y] + \delta,
\end{align*}
which is equivalent to 
\begin{align*}
\Pr[f(\mathsf{sw}_1(t_i, \vec{t}'_I, \vec{t}_J), \ldots, \mathsf{sw}_m(t_i, \vec{t}'_I, \vec{t}_J)) \in Y] \leq e^{\eps'} \Pr[f(\mathsf{sw}_1(t_i', \vec{t}'_I, \vec{t}_J), \ldots, \mathsf{sw}_m(t_i', \vec{t}'_I, \vec{t}_J)) \in Y] + \delta
\end{align*}
for some function $f: \R^m \to S$. It is easy to see that it suffices to show that for every measurable set $Y' \subseteq \R^m$, 
\begin{align*}
\Pr[(\mathsf{sw}_1(t_i, \vec{t}_J), \ldots, \mathsf{sw}_m(t_i, \vec{t}_J)) \in Y'] \leq e^{\eps'} \Pr[(\mathsf{sw}_1(t_i', \vec{t}_J), \ldots, \mathsf{sw}_m(t_i', \vec{t}_J)) \in Y'] + \delta. \tag{1}
\end{align*}
Now, we note that the random variables $\{\sw_\ell(t_i,\vec{t}_J) : \ell \in [m]\}$ are mutually independent. Thus, using known ``composition theorems/results'' for differential privacy (e.g., see \cite{DKMMN06,DL09,DRV10}), it is not hard to see that to prove (1), it suffices to show that for every $\ell \in [m]$ and every measurable set $Y' \subseteq \R$, we have
\begin{align*}
\Pr[\mathsf{sw}_\ell(t_i, \vec{t}_J) \in Y'] \leq e^{\eps'/m} \Pr[\mathsf{sw}_\ell(t_i', \vec{t}_J) \in Y'] + \delta/m. \tag{2}
\end{align*}
To this end, fix $\ell \in [m]$. Given a player $j$ with reported type $t_j$ and weight $w_{j,\ell}$ for option $o_\ell$, let
\begin{align*}
h(t_j) = \frac{\sw_\ell(t_j) - w_{j,\ell} \cdot \E_{X \sim \T_\ell}[X]}{(\Var_{X \sim \T_\ell}[X])^{1/2}},
\end{align*} 
and given a subset $A$ of players with reported types $\vec{t}_A$, let 
\begin{align*}
V(\vec{t}_A) = \sum_{j \in A} h(t_j). 
\end{align*}
Let $J' = \{j \in J : w_{j,\ell} = 1\}$, and let $n' = |J'|$. We note that for every $j \in J'$, we have $\E[h(t_j)] = 0$ and $\Var[h(t_j)] = 1$. To show (2), it is easy to see that it suffices to show that for every measurable set $W \subseteq \R$,
\begin{align*}
\Pr\left[\frac{V(t_i, \vec{t}_{J'})}{\sqrt{n'}} \in W\right] \leq e^{\epsilon'/m} \Pr\left[\frac{V(t_i', \vec{t}_{J'})}{\sqrt{n'}} \in W\right] + \delta/m. \tag{3}
\end{align*}
To this end, fix a measurable set $W \subseteq \R$. We will need to use the following lemma later:

\begin{lemma}
Let $f_{n'}$ be the pdf of $\frac{V(\vec{t}_{J'})}{\sqrt{n'}} = \frac{1}{\sqrt{n'}} \sum_{j \in {J'}} h(t_j)$ (where each $t_j \sim \T$ independently), and let $\phi$ be the pdf of the standard normal distribution. Then, $f_{n'}$ converges uniformly to $\phi$ as follows: 
\begin{align*}
\sup_{x \in \R} |f_{n'}(x) - \phi(x)| \leq \frac{c_1}{\sqrt{n'}},
\end{align*}
where $c_1 > 0$ is some (universal) constant.
\end{lemma}

\begin{proof}[Proof of lemma]
This lemma follows immediately from various ``local limit theorems'' (e.g., see Statement 5 in Section 4 of Chapter VII in \cite{Pet75}).
\end{proof}

Let $\eps'' > 0$. Let $c_1$ be the constant in the above lemma, let $\gamma = \frac{e^{\eps''}-1}{e^{\eps''}+1}$, and let $\sigma^2 = \Var_{X \sim \T_\ell}[X]$. Let 
\begin{align*}
B = \left\{w \in \R : |w| > \sqrt{2\ln\left(\frac{\gamma \sqrt{n'}}{c_1\sqrt{2\pi}}\right)} - \frac{4\alpha}{\sigma \sqrt{n'}}\right\}.
\end{align*}
Now, observe that
\begin{align*}
\Pr\left[\frac{V(t_i,\vec{t}_{J'})}{\sqrt{n'}} \in W\right] 
\leq \Pr\left[\frac{V(t_i,\vec{t}_{J'})}{\sqrt{n'}} \in W \cap \overline{B}\right] + \Pr\left[\frac{V(t_i,\vec{t}_{J'})}{\sqrt{n'}} \in B\right]. \tag{4}
\end{align*}

\begin{lemma}
\label{lem:manyChoicesCase1}
For each $w \in W \cap \overline{B}$,
\begin{align*}
\frac{\Pr[V(t_i, \vec{t}_{J'}) / \sqrt{n'} = w]}{\Pr[V(t_i', \vec{t}_{J'}) / \sqrt{n'} = w]} \leq \exp\left(\eps'' + O\left(\frac{\sqrt{\ln n'}}{\sqrt{n'}}\right)\right).
\end{align*}
\end{lemma}

\begin{proof}[Proof of lemma]
Fix $w \in W \cap \overline{B}$. Let $\phi$ be the pdf of the standard normal distribution. By the previous lemma, we have
\begin{align*}
\frac{\Pr[V(t_i, \vec{t}_{J'}) / \sqrt{n'} = w]}{\Pr[V(t_i', \vec{t}_{J'}) / \sqrt{n'} = w]}
= \ & \frac{\Pr[\frac{1}{\sqrt{n'}} \sum_{j \in {J'}} h(t_j) = w - h(t_i) / \sqrt{n'}]}{\Pr[\frac{1}{\sqrt{n'}} \sum_{j \in {J'}} h(t_j) = w - h(t_i') / \sqrt{n'}]} \\
\leq \ & \frac{\phi(w-h(t_i)/\sqrt{n'}) + \frac{c_1}{\sqrt{n'}}}{\phi(w-h(t_i')/\sqrt{n'}) - \frac{c_1}{\sqrt{n'}}}, \tag{5}
\end{align*}
where $c_1 \geq 0$ is some constant. Now, we note that
\begin{align*}
\frac{c_1}{\sqrt{n'}} \leq \gamma \cdot \phi(w - h(t_i)/\sqrt{n'}), \tag{6}
\end{align*}
since $\gamma \cdot \phi(w - h(t_i)/\sqrt{n'}) = \frac{\gamma}{\sqrt{2\pi}} \exp\left(-\frac{1}{2}\left(w - \frac{h(t_i)}{\sqrt{n'}}\right)^2\right) \geq \frac{\gamma}{\sqrt{2\pi}} \exp\left(-\frac{1}{2}\left(|w| + \frac{|h(t_i)|}{\sqrt{n'}}\right)^2\right)$, $|h(t_i)| \leq \frac{4\alpha}{\sigma}$, and $|w| + \frac{4\alpha}{\sigma \sqrt{n'}} \leq \sqrt{2\ln\left(\frac{\gamma \sqrt{n'}}{c_1\sqrt{2\pi}}\right)}$ (since $w \notin B$). Similarly, we also have
\begin{align*}
\frac{c_1}{\sqrt{n'}} \leq \gamma \cdot \phi(w - h(t_i')/\sqrt{n'}). \tag{7}
\end{align*}
Now, combining (5) with (6) and (7), we have 
\begin{align*}
\frac{\Pr[V(t_i, \vec{t}_{J'}) / \sqrt{n'} = w]}{\Pr[V(t_i', \vec{t}_{J'}) / \sqrt{n'} = w]}
\leq \ & \frac{\phi(w-h(t_i)/\sqrt{n'}) + \gamma \cdot \phi(w - h(t_i)/\sqrt{n'})}{\phi(w-h(t_i')/\sqrt{n'}) - \gamma \cdot \phi(w - h(t_i')/\sqrt{n'})} \\
= \ & \frac{1+\gamma}{1-\gamma} \cdot \frac{\phi(w-h(t_i)/\sqrt{n'})}{\phi(w-h(t_i')/\sqrt{n'})} \\
= \ & \exp(\eps'') \cdot \exp\left(\frac{1}{2} \cdot \left(\left(w-\frac{h(t_i')}{\sqrt{n'}}\right)^2 - \left(w-\frac{h(t_i)}{\sqrt{n'}}\right)^2\right)\right) \\
\leq \ & \exp(\eps'') \cdot \exp\left(\frac{1}{2} \cdot \left|\frac{h(t_i)}{\sqrt{n'}} - \frac{h(t_i')}{\sqrt{n'}}\right| \cdot \left|2w - \frac{h(t_i')}{\sqrt{n'}} - \frac{h(t_i)}{\sqrt{n'}}\right|\right) \\
\leq \ & \exp(\eps'') \cdot \exp\left(\frac{1}{2} \cdot \frac{8\alpha}{\sigma\sqrt{n'}} \cdot \left(2 \sqrt{2 \ln\left(\frac{\gamma\sqrt{n'}}{c_1\sqrt{2\pi}}\right)} \right)\right) \\
\leq \ & \exp\left(\eps'' + O\left(\frac{\sqrt{\ln n'}}{\sqrt{n'}}\right)\right),
\end{align*}
where the second last inequality follows from the fact that the function $h$ is bounded by $\frac{4\alpha}{\sigma}$ (in absolute value) and $|w| + \frac{4\alpha}{\sigma\sqrt{n'}} \leq \sqrt{2 \ln\left(\frac{\gamma\sqrt{n'}}{c_1\sqrt{2\pi}}\right)}$ (since $w \notin B$).
\end{proof}

We will need to use the following lemma later:

\begin{lemma}
Let $F_{n'}$ be the cdf of $\frac{V(\vec{t}_{J'})}{\sqrt{n'}} = \frac{1}{\sqrt{n'}} \sum_{j \in {J'}} h(t_j)$ (where each $t_j \sim \T$ independently), and let $\Phi$ be the cdf of the standard normal distribution. Then, $\frac{V(\vec{t}_{J'})}{\sqrt{n'}}$ converges in distribution uniformly to the standard normal random variable as follows: 
\begin{align*}
\sup_{x \in \R} |F_{n'}(x) - \Phi(x)| = O\left(\frac{1}{\sqrt{n'}}\right).
\end{align*}
\end{lemma}

\begin{proof}[Proof of lemma]
This lemma follows immediately from the Berry-Esseen theorem.
\end{proof}

\begin{lemma}
\label{lem:manyChoicesCase2}
\begin{align*}
\Pr\left[\frac{V(t_i,\vec{t}_{J'})}{\sqrt{n'}} \in B\right] \leq O\left(\frac{1}{\eps'' \sqrt{n'}}\right).
\end{align*}
\end{lemma}

\begin{proof}[Proof of lemma]
Let $x = \sqrt{2\ln\left(\frac{\gamma \sqrt{n'}}{c_1\sqrt{2\pi}}\right)} - \frac{4\alpha}{\sigma \sqrt{n'}}$, let $\Phi$ be the cdf of the standard normal distribution, and let $X \sim \mathcal{N}(0,1)$. Observe that
\begin{align*}
\Pr\left[\frac{V(t_i,\vec{t}_{J'})}{\sqrt{n'}} \in B\right] 
\leq \ & \Pr\left[\left|\frac{h(t_i)}{\sqrt{n'}}\right| + \left|\frac{V(\vec{t}_{J'})}{\sqrt{n'}}\right| > x \right] \\
\leq \ & 1 - \Pr\left[\left|\frac{V(\vec{t}_{J'})}{\sqrt{n'}}\right| \leq x - \frac{4\alpha}{\sigma \sqrt{n'}} \right] \\
\leq \ & 1 - \left(\Pr\left[\frac{V(\vec{t}_{J'})}{\sqrt{n'}} \leq x - \frac{4\alpha}{\sigma \sqrt{n'}} \right] - \Pr\left[\frac{V(\vec{t}_{J'})}{\sqrt{n'}} \leq - \left(x - \frac{4\alpha}{\sigma \sqrt{n'}}\right)\right]\right) \\
\leq \ & 1 - \left(\Phi\left(x - \frac{4\alpha}{\sigma \sqrt{n'}}\right) - \Phi\left(-\left(x - \frac{4\alpha}{\sigma \sqrt{n'}}\right)\right) - O\left(\frac{1}{\sqrt{n'}}\right)\right) \\
= \ & 1 - \Pr\left[|X| \leq x - \frac{4\alpha}{\sigma \sqrt{n'}}\right] + O\left(\frac{1}{\sqrt{n'}}\right) \\
= \ & \Pr\left[|X| > x - \frac{4\alpha}{\sigma \sqrt{n'}}\right] + O\left(\frac{1}{\sqrt{n'}}\right),
\end{align*}
where the second inequality follows from the fact that $|h(t_i)| \leq \frac{4\alpha}{\sigma}$, and the fourth inequality follows from the previous lemma. Now, by Mill's inequality, we have $\Pr[|X| \geq y] \leq \sqrt{\frac{2}{\pi}} \cdot \frac{1}{y} e^{-y^2/2}$ for every $y > 0$, so
\begin{align*}
\Pr\left[\frac{V(t_i,\vec{t}_{J'})}{\sqrt{n'}} \in B\right] 
\leq \ & \sqrt{\frac{2}{\pi}} \cdot \frac{1}{x - \frac{4\alpha}{\sigma \sqrt{n'}}} \exp\left(-\frac{1}{2} \left(x - \frac{4\alpha}{\sigma \sqrt{n'}}\right)^2 \right) + O\left(\frac{1}{\sqrt{n'}}\right) \\
= \ & O\left(\frac{1}{\gamma \sqrt{n'}}\right) \\
= \ & O\left(\frac{1}{\eps'' \sqrt{n'}}\right).
\end{align*}
\end{proof}

Now, by Lemmas \ref{lem:manyChoicesCase1} and \ref{lem:manyChoicesCase2}, we have from (4) that
\begin{align*}
\Pr\left[\frac{V(t_i,\vec{t}_{J'})}{\sqrt{n'}} \in W\right] 
\leq \ & \Pr\left[\frac{V(t_i,\vec{t}_{J'})}{\sqrt{n'}} \in W \cap \overline{B}\right] + \Pr\left[\frac{V(t_i,\vec{t}_{J'})}{\sqrt{n'}} \in B\right] \\
\leq \ & \exp\left(\eps'' + O\left(\frac{\sqrt{\ln n'}}{\sqrt{n'}}\right)\right) \cdot \Pr\left[\frac{V(t_i',\vec{t}_{J'})}{\sqrt{n'}} \in W \cap \overline{B}\right] + O\left(\frac{1}{\eps'' \sqrt{n'}}\right) \\
\leq \ & \exp\left(\eps'' + O\left(\frac{\sqrt{\ln n}}{\sqrt{n}}\right)\right) \cdot \Pr\left[\frac{V(t_i',\vec{t}_{J'})}{\sqrt{n}} \in W\right] + O\left(\frac{1}{\eps'' \sqrt{n}}\right),
\end{align*}
where the last inequality follows from the fact that $n' = \Theta(n)$. Now, set $\eps'' = \eps / m$ so that
\begin{align*}
\Pr\left[\frac{V(t_i,\vec{t}_{J'})}{\sqrt{n'}} \in W\right] 
\leq \ & \exp\left(\frac{\eps}{m} + O\left(\frac{\sqrt{\ln n}}{\sqrt{n}}\right)\right) \cdot \Pr\left[\frac{V(t_i',\vec{t}_{J'})}{\sqrt{n}} \in W\right] + O\left(\frac{m}{\eps \sqrt{n}}\right) \\
\leq \ & \exp\left(\frac{\epsilon'}{m}\right) \Pr\left[\frac{V(t_i', \vec{t}_{J'})}{\sqrt{n'}} \in W\right] + \frac{\delta}{m}.
\end{align*}
This proves (3), which concludes the proof of the theorem.
\end{proof}

\section{Obtaining Actual Truthfulness \label{sec:punishing}}

In this section, we present a simple method of obtaining (actual) truthfulness from approximate truthfulness.
Approximate truthfulness, being an intermediate state between the complete lack of truthfulness and actual truthfulness, is a natural candidate to be a ``stepping stone'' for obtaining truthfulness.
For example, in \cite{NST12}, Nissim et al. consider mechanisms with an additional ``reaction stage'', in which the social alternative has already been chosen and each player chooses a reaction (for example, in a facility location problem a reaction would be going to a particular facility).
The mechanism can restrict the set of available reactions to a player based on his/her reported type (for example, in a facility location problem a player can only go to the facility that is closest to his/her reported type). 
Intuitively, if having a different type means that the player will have a different optimal reaction to the chosen social alternative, then the mechanism designer can make liars suffer by only allowing each player to choose the reaction that is optimal for his/her reported type. 
This is the idea behind the ``imposing mechanism'' in \cite{NST12}. 
The imposing mechanism chooses a social alternative randomly, ignoring the agents' reported types, and then restricts the players' reactions such that for every player, lying results in a loss of at least $\eps$ for some $\eps > 0$. 

With the ``imposing mechanism'' as a tool, Nissim et al. proceed to show that an appropriate randomization between the imposing mechanism and an approximately truthful mechanism that has good efficiency results in a truthful mechanism that also has good efficiency.
We can also use a construction similar to the imposing mechanism to obtain persistent Bayes-Nash truthful mechanisms from persistent $\eps$-Bayes-Nash truthful mechanisms. 
However, since the imposing mechanism disregards the players' reported types when choosing a social alternative, truthfulness comes at a cost of decreased efficiency (this is also the case in \cite{NST12}).

We propose a simple alternative technique to obtain truthfulness from $\eps$-truthfulness, without sacrificing any efficiency of the mechanism. 
In our model, the mechanism can verify the truthfulness of a small number of players, and impose a fine on any player caught lying.
We call this part of the mechanism the \emph{deterrent payment scheme}.
Although the assumption that a reported type can be verified does not hold in settings such as political elections, where player preferences are subjective, there are many other settings in which this assumption is appropriate.

One might ask, if the truth of the reported types can be verified, why doesn't the mechanism just find out all the true types, making the mechanism design problem trivial?
The first reason is that verification is often costly. 
For example, the government cannot afford to check every income tax report, but audits are performed to deter individuals and organizations from cheating on their taxes.
Secondly, verifying the truth of an input is often easier than determining the true value.
For example, an organization may not be able to obtain an individual's address from his or her name, but it can send someone to visit a provided address to verify whether the individual actually lives there. 
Similarly, verification emails are routinely used to check whether an individual actually has access to the email address he or she provided.

We note that the deterrent payment scheme is much more meaningful and realistic when it is used with an $\eps$-truthful mechanism, than when it is used with general mechanisms.
When a mechanism is already $\eps$-truthful, neither the number of verifications nor the fine needs to be large in order to get actual truthfulness. 
Without $\eps$-truthfulness, either the number of verifications may need to be too large to be affordable, or the fine may need to be too large to be enforceable.

Finally, although we need to assume quasilinear utilities to use the deterrent payment scheme, and the deterrent payment scheme involves a transfer of utilities (i.e. payments), only those who lie will ever be required to pay. 
This means that in the truthful equilibrium, no payments are made and the mechanism will not have excess money that it might need to `burn'. 
The fact that truthful players never pay, and the fact that individual rationality in the truthful equilibrium is unaffected by a deterrent payment scheme, means that a deterrent payment scheme may be useful in some settings where the AGV mechanism by Arrow \cite{Arr79} and d'Aspremont and G\'{e}rard-Varet \cite{AGV79} may be inappropriate.

\subsection{Deterrent Payment Scheme}

In this section, we assume that each player $i$ has a quasilinear utility function $\wt{u}_i : T \times S \times \R \rightarrow \R$ such that $\wt{u}_i(t,s,p) = u_i(t,s) - p$, where $u_i$ is the utility function from Section~\ref{sec:preliminaries}. 
The notions of truthfulness defined in Section~\ref{sec:truthfulness} can be straight-forwardly adapted to the quasilinear utility setting, and references to truthfulness in this section refer to these adapted notions. 

\begin{definition}[\bf{Deterrent Payment Scheme}]
A \emph{deterrent payment scheme} with $m$ verifications is a function $P : T^n \times ([n]\times \{0,1\})^m \rightarrow \R^n$. 
\end{definition}
Intuitively, the first input to the payment scheme is a vector of reported types; the second input to the payment scheme is a set of $m$ pairs, where a pair $(i,1)$ indicates that player $i$ is lying, and $(i,0)$ indicates that player $i$ is not lying. 
In general, this input does not provide the true types of the players, so even if the verification includes all players, the mechanism still will not know the vector of true types.

We will show that the deterrent payment scheme can be added to a persistent approximately Bayes-Nash truthful mechanism to make the resulting mechanism Bayes-Nash truthful. 
In particular, we will show that truthfulness is a \emph{weakly persistent Bayes-Nash equilibrium}. 
Weakly persistent Bayes-Nash equilibrium is a weakening of the persistent Bayes-Nash equilibrium.
Instead of guaranteeing that no individual gains by being in a coalition, it guarantees that any coalition as a whole cannot gain by deviating from the equilibrium. 
This means there are no side-payments that can be made between members in a coalition such that each member gains by being in the coalition.
We make the notion precise in the following definition.
%Intuitively, a $(k,r)$-persistent $\eps$-Bayes-Nash truthful mechanism with a deterrent payment scheme does not achieve the stronger notion of $(k,r)$-persistent Bayes-Nash truthfulness, since a player $i$ can gain up to $\eps$ by participating in a coalition but still reporting his true type (and hence never getting fined). 

\begin{definition}[{\bf Weakly $(k,r)$-persistent Bayes-Nash truthful mechanism with a deterrent payment scheme}]\label{def:weakPersistent}
A mechanism ${M}$ with deterrent payment scheme $P$ with $m$ verifications is weakly $(k,r)$-persistent Bayes-Nash truthful if for every $I \subseteq [n]$ with $|I| \leq k$, every possible announced types $\vec{t}'_I \in T^{|I|}$ for $I$, every coalition $C \subseteq [n] \setminus I$ with $|C| \leq r$, every true types $\vec{t}_C \in T^{|C|}$ for the coalition $C$, and every possible announced types $\vec{t}'_C \in T^{|C|}$ for the coalition $C$, we have 
\begin{align*}
&\sum_{i \in C}\E_{\vec{t}_J,X}[u_i(t_i,{M}(\vec{t}_C, \vec{t}'_I, \vec{t}_{J})) - P_i((\vec{t}_C, \vec{t}'_I, \vec{t}_{J}), X)] \\
\geq \ & \sum_{i \in C} \E_{\vec{t}_{J},X}[u_i(t_i,{M}(\vec{t}'_C, \vec{t}'_I, \vec{t}_{J})) - P_i((\vec{t}'_C, \vec{t}'_I, \vec{t}_{J}), X)],
\end{align*}
where $J = [n] \setminus (I \cup C)$, $\vec{t}_J \sim \T^{|J|}$, and $X$ is a random vector of $m$ pairs from $[n]\times \{0,1\}$, obtained by uniform-randomly selecting $m$ players $\ell$ without replacement from $[n]$ and including $(\ell,0)$ if player $j$ was truthful (the announced type $\hat{t}_\ell = t_\ell$), or including $(\ell,1)$ if player $\ell$ was not truthful. 
\end{definition}

Note that in Definition~\ref{def:weakPersistent}, truthfulness is defined with respect to uniformly sampling and verifying $m$ of the players.
This means that the deterrent payment scheme does not choose which players will be verified. 

We can define \emph{$k$-tolerant Bayes-Nash truthfulness} for mechanisms with deterrent payment schemes as weakly $(k,1)$-persistent Bayes-Nash truthfulness. 
The term ``weakly'' no longer applies when we consider only $k$-tolerance, which does not involve coalitions.

Next, we observe that by pairing a $\epsilon$-Bayes-Nash truthful mechanism with a simple deterrent payment scheme that imposes a fixed fine on any individual that is caught lying, we obtain a persistent Bayes-Nash truthful mechanism. 

\begin{theorem}\label{thm:deterrent}
Let $M$ be any $(k,r)$-persistent $\epsilon$-Bayes-Nash truthful mechanism, and let $P$ be a deterrent payment scheme with $m$ verifications defined by 
\begin{align*}
P_i( \vec{t}, V ) = 
\begin{cases}
0 \text{ if } (i,1) \notin V \\
d \text{ if } (i,1) \in V.
\end{cases} 
\end{align*}
Then $M$ with deterrent payment scheme $P$ is
\begin{description}
  \item[\quad (1)] weakly $(k,r)$-persistent Bayes-Nash truthful, if $\frac{m}{n} \cdot d \geq r\eps$,
  \item[\quad (2)]$k$-tolerant Bayes-Nash truthful, if $\frac{m}{n} \cdot d \geq \eps$.
\end{description}
\end{theorem}

\begin{proof}
By definition of $(k,r)$-persistent $\epsilon$-Bayes-Nash truthfulness of $M$,
for every $I \subseteq [n]$ with $|I| \leq k$, every possible announced types $\vec{t}'_I \in T^{|I|}$ for $I$, every coalition $C \subseteq [n] \setminus I$ with $|C| \leq r$, every true types $\vec{t}_C \in T^{|C|}$ for the coalition $C$, every player $i \in C$ in the coalition, and every possible announced types $\vec{t}'_C \in T^{|C|}$ for the coalition $C$, we have 
\begin{align*}
\E_{\vec{t}_J}[u_i(t_i,M(\vec{t}_C, \vec{t}'_I, \vec{t}_{J}))] \geq \E_{\vec{t}_{J}}[u_i(t_i,M(\vec{t}'_C, \vec{t}'_I, \vec{t}_{J}))] - \epsilon,
\end{align*}
where $J = [n] \setminus (I \cup C)$ and $\vec{t}_J \sim \T^{|J|}$.

Note that the expected fine for any player $\ell$ in coalition $C$ for deviating is 
\begin{align*}
\E_{\vec{t}_J,X}[P_\ell((\vec{{t}}_C', \vec{t}'_I, \vec{t}_{J}),X)] = 
\begin{cases}
0 \text{ if } {t}_\ell' = t_\ell \\
\frac{m}{n} \cdot d \text{ if } {t}_\ell' \neq t_\ell.
\end{cases}
\end{align*}

We will first show part \textbf{(1)}. 
For any $\vec{t}'_C \neq \vec{t}_C$, at least one player $\ell$ in $C$ must be lying, and thus have an expected fine of $\frac{m}{n} \cdot d$.

Since $\frac{m}{n} \cdot d \geq |C|\eps$, we have 
\begin{align*}
& \sum_{i \in C}\E_{\vec{t}_J,X}[u_i(t_i,{M}(\vec{t}_C, \vec{t}'_I, \vec{t}_{J})) - P_i((\vec{t}_C, \vec{t}'_I, \vec{t}_{J}),X)] \\
= \ &  \sum_{i \in C}\E_{\vec{t}_J}[u_i(t_i,{M}(\vec{t}_C, \vec{t}'_I, \vec{t}_{J}))] \\
\geq \ &  \sum_{i \in C}\left(  \E_{\vec{t}_{J}}[u_i(t_i,{M}(\vec{t}'_C, \vec{t}'_I, \vec{t}_{J}))] - \eps\right)\\
\geq \ &  \sum_{i \in C} \E_{\vec{t}_{J},X}[u_i(t_i,{M}(\vec{t}'_C, \vec{t}'_I, \vec{t}_{J})) - P_i((\vec{t}_C, \vec{t}'_I, \vec{t}_{J}),X)].
\end{align*}
To see that \textbf{(2)} holds, note that by part \textbf{(1)}, ${M}$ with payment scheme $P$ is weakly $(k,1)$-persistent Bayes-Nash truthful, and hence $k$-tolerant Bayes-Nash truthful, if $\frac{m}{n} \cdot d \geq \eps$.
\end{proof}

We conclude this section with two examples that illustrate the use of a simple deterrent payment scheme with the mechanisms from Section~\ref{sec:mechs}.
\begin{example}[\bf{Facility Location}]
In Example~\ref{ex:facility}, starting with a $(k,r)$-persistent $\epsilon$-Bayes-Nash truthful mechanism, the mechanism designer can discourage lying by randomly sampling $m$ individuals, and imposing a fine on each individual caught lying about their residence. 
For instance, if $r=10$ and the fine is fixed at $10000 \times \epsilon$, then the mechanism need only verify $0.1\%$ of the players to ensure truthfulness.
\end{example}

\begin{example}[\bf{Golf Course vs. Low Cost Swimming Pool}]
In Example~\ref{ex:golf}, to discourage members from misreporting their income, the club will randomly select $m$ of the members and ask the member to provide his/her tax forms from the previous year. 
If the member is found to have lied about his/her income, the club will impose a fine on the member. 
%If the number of verifications and the amount of fine is large enough, then by Theorem~\ref{thm:deterrent}, adding this deterrent payment scheme to the mechanism changes it from being only approximately truthful to be actually truthful. 
\end{example}

\bibliographystyle{amsalpha}
\bibliography{MechanismDesignPrivacy}

\appendix
\renewcommand\thesection{Appendix \Alph{section}}

\section{Proofs of Basic Theorems about Privacy and Truthfulness}
\label{app:privacyTruthfulness}

We first state a group version of Bayesian differential privacy that protects the privacy of groups of at most $c$ players.

\begin{definition}[{\bf $(c,k,\eps,\delta)$-Bayesian group differential privacy}]
A mechanism $M: T^n \to S$ is \emph{$(c,k,\eps,\delta)$-Bayesian group differentially private} if for every group of players $C \subseteq [n]$ with $|C| \leq c$, every subset $I \subseteq [n] \setminus C$ of players with $|I| \leq k$, every pair of types $\vec{t}_C, \vec{t}_C' \in T^{|C|}$ for the players in $C$, and every $\vec{t}'_I \in T^{|I|}$, the following holds: Let $J = [n] \setminus (I \cup C)$ (the remaining players), and $\vec{t}_J \sim \T^{|J|}$; then, for every $Y \subseteq S$, we have
\begin{align*}
\Pr[M(\vec{t}_C, \vec{t}'_I, \vec{t}_J) \in Y] \leq e^\epsilon \cdot \Pr[M(\vec{t}_C', \vec{t}'_I, \vec{t}_J) \in Y] + \delta, 
\end{align*}
where the probabilities are over $\vec{t}_J \sim \T^{|J|}$.
\end{definition}

We now prove that Bayesian differential privacy implies group Bayesian differential privacy, but the parameters degrade with increasing group size.

\begin{theorem}[{\bf Bayesian differential privacy $\implies$ Group Bayesian differential privacy}]
\label{thm:BDP=>GBDP}
Let $M: T^n \to S$ be any mechanism that is $(k,\eps,\delta)$-Bayesian differentially private. Then, for every $1 \leq c \leq k+1$, $M$ is also $(c,k-c+1,c\eps,\frac{e^{c\eps}-1}{e^\eps-1}\delta)$-Bayesian group differentially private.
\end{theorem}

Intuitively, if a mechanism is $(k,\eps,\delta)$-Bayesian differentially private, then when a player $i$ changes his/her announced type to something non-truthful, the output distribution of the mechanism changes very little, so player $i$'s utility also changes very little. This holds even if $k$ players are non-truthful, so the mechanism is $k$-tolerant approximate Bayes-Nash truthful. To obtain resilient and persistent approximate Bayes-Nash truthfulness, we show that $(k,\eps,\delta)$-Bayesian differential privacy implies a \emph{group} version of $(k,\eps,\delta)$-Bayesian differential privacy that protects the privacy of groups of players of a certain size (but the parameters $\eps$ and $\delta$ degrade with the group size). Group privacy would then imply that even if a coalition of players change their announced types, the output distribution of the mechanism changes very little, so the utility of each player in the coalition changes very little. Thus, the mechanism is persistent approximate Bayes-Nash truthful.

\begin{proof}
Let $1 \leq c \leq k+1$, let $C \subseteq [n]$ with $|C| \leq c$, let $I \subseteq [n] \setminus C$ with $|I| \leq k-c+1$, let $\vec{t}_C, \vec{t}_C' \in T^{|C|}$, let $\vec{t}'_{I} \in T^{|I|}$, and let $Y \subseteq S$. Let $J = [n] \setminus (I \cup C)$ and $\vec{t}_J \sim \T^{|J|}$. We need to show that
\begin{align*}
  \Pr[M(\vec{t}_C, \vec{t}'_{I}, \vec{t}_J) \in Y] \leq e^{c\epsilon} \Pr[M(\vec{t}_C', \vec{t}'_{I}, \vec{t}_J) \in Y] + \frac{e^{c\eps}-1}{e^\eps-1} \cdot \delta. \tag{1}
\end{align*}
Since $M$ is $(k,\eps,\delta)$-Bayesian differentially private, for every $i \in [n]$, every $I' \subseteq [n] \setminus \{i\}$ with $|I'| \leq k$, every $t_i, t_i' \in T$, and every $\vec{t}'_{I'} \in T^{|I'|}$, we have
\begin{align*}
  \Pr[M(t_i, \vec{t}'_{I'}, \vec{t}_{J'}) \in Y] \leq e^\epsilon \Pr[M(t_i', \vec{t}'_{I'}, \vec{t}_J) \in Y] + \delta, \tag{2}
\end{align*}
where $J' = [n] \setminus (I' \cup \{i\})$ and $\vec{t}_{J'} \sim \T^{|J'|}$.

Suppose $C = \{i_1, \ldots, i_{|C|}\}$, $\vec{t}_C = (t_{i_1}, \ldots, t_{i_{|C|}})$, and $\vec{t}_C' = (t_{i_1}', \ldots, t_{i_{|C|}}')$. Let $\vec{t}_C^{(0)} = \vec{t}_C$ and $\vec{t}_C^{(|C|)} = \vec{t}_C'$, and more generally, for $j = 0, \ldots, |C|$, let 
\begin{align*}
\vec{t}_C^{(j)} = (t_{i_1}', \ldots, t_{i_j}', t_{i_{j+1}}, \ldots, t_{i_{|C|}}).
\end{align*}

\begin{claim}
For every $j \in \{0, \ldots, |C|-1\}$, 
\begin{align*}
\Pr[M(\vec{t}_C^{(j)}, \vec{t}'_{I}, \vec{t}_J) \in Y] \leq e^\epsilon \Pr[M(\vec{t}_C^{(j+1)}, \vec{t}'_{I}, \vec{t}_J) \in Y] + \delta. \tag{3}
\end{align*}
\end{claim}

\begin{proof}[Proof of claim]
Fix $j \in \{0, \ldots, |C|-1\}$. Let $i = i_{j+1}$, $I' = (C \setminus \{i\}) \cup I$, $t_i = t_{i_{j+1}}$, $t_i' = t_{i_{j+1}}'$, and $\vec{t}_{I'}' = ((\vec{t}_C^{(j)})_{-i}, \vec{t}_{I}')$. Let $J' = [n] \setminus (I' \cup \{i\})$ and $\vec{t}_{J'} \sim \T^{|J'|}$. Then, by (2), we have
\begin{align*}
  \Pr[M(t_i, \vec{t}'_{I'}, \vec{t}_{J'}) \in Y] \leq e^\epsilon \Pr[M(t_i', \vec{t}'_{I'}, \vec{t}_{J'}) \in Y] + \delta, 
\end{align*}
which is equivalent to 
\begin{align*}
  \Pr[M(\vec{t}_C^{(j)}, \vec{t}'_{I}, \vec{t}_J) \in Y] \leq e^\epsilon \Pr[M(\vec{t}_C^{(j+1)}, \vec{t}'_{I}, \vec{t}_J) \in Y] + \delta,
\end{align*}
as required. 
\end{proof}

Using the above claim repeatedly starting with $j = 0$, we get
\begin{align*}
\Pr[M(\vec{t}_C^{(0)}, \vec{t}'_{I}, \vec{t}_J) \in Y] \leq e^{|C| \epsilon} \Pr[M(\vec{t}_C^{(|C|)}, \vec{t}'_{I}, \vec{t}_J) \in Y] + \frac{e^{|C|\eps}-1}{e^\eps-1} \cdot \delta,
\end{align*}
which yields (1) since $|C| \leq c$, as required.
\end{proof}

We now prove Theorem \ref{thm:BDP=>Truthfulness}, which we restate here for convenient reference.

\begin{reptheorem}{thm:BDP=>Truthfulness}[\bf Bayesian differential privacy $\implies$ Persistent approximate truthfulness]
Suppose the utility functions are bounded by $\alpha > 0$, i.e., the utility function for each player $i$ is $u_i: T \times S \to [-\alpha,\alpha]$. Let $M$ be any mechanism that is $(k,\eps,\delta)$-Bayesian differentially private. Then, $M$ satisfies the following properties:
\begin{enumerate}
  \item $M$ is $k$-tolerant $(\eps + 2\delta)(2\alpha)$-Bayes-Nash truthful.
  \item For every $1 \leq r \leq k+1$, $M$ is $r$-resilient $(r\eps + 2r\delta)(2\alpha)$-Bayes-Nash truthful.
  \item For every $1 \leq r \leq k+1$, $M$ is $(k-r+1,r)$-persistent $(r\eps + 2r\delta)(2\alpha)$-Bayes-Nash truthful.
\end{enumerate}
\end{reptheorem}

\begin{proof}
Property 1 follows from Property 3 by setting $r = 1$ in Property 3. Property 2 follows immediately from Property 3. Thus, it suffices to only show Property 3. 

Let $1 \leq r \leq k+1$. We first note that if $r \eps \geq 1$, then Property 3 immediately holds, since the utility functions are bounded by $\alpha$. Thus, we now assume that $r \eps < 1$. By Theorem \ref{thm:BDP=>GBDP}, since $M$ is $(k,\eps,\delta)$-Bayesian differentially private, $M$ is also $(r,k-r+1,r\eps,\frac{e^{r\eps}-1}{e^\eps-1}\delta)$-Bayesian group differentially private. One can easily verify that $e^{r\eps} - 1 \leq 2 r\eps$ for $r \eps < 1$, so $\frac{e^{r\eps}-1}{e^\eps-1} \leq \frac{2r\eps}{\eps} = 2r$. Thus, $M$ is $(r,k-r+1,r\eps,2r\delta)$-Bayesian group differentially private. Then, for every $C \subseteq [n]$ with $|C| \leq r$, every $I \subseteq [n] \setminus C$ with $|I| \leq k-r+1$, every $\vec{t}_C, \vec{t}_C' \in T^{|C|}$, every $\vec{t}'_I \in T^{|I|}$, and every $Y \subseteq S$, we have
\begin{align*}
\Pr[M(\vec{t}_C, \vec{t}'_I, \vec{t}_J) \in Y] \leq e^{r\epsilon} \Pr[M(\vec{t}_C', \vec{t}'_I, \vec{t}_J) \in Y] + 2r\delta, \tag{1}
\end{align*}
where $J = [n] \setminus (I \cup C)$ and $\vec{t}_J \sim \T^{|J|}$.

Let $I \subseteq [n]$ with $|I| \leq k-r+1$, let $\vec{t}'_I \in T^{|I|}$, let $C \subseteq [n] \setminus I$ with $|C| \leq r$, let $\vec{t}_C \in T^{|C|}$, let $i \in C$, and let $\vec{t}'_C \in T^{|C|}$. Let $J = [n] \setminus (I \cup C)$ and $\vec{t}_J \sim \T^{|J|}$. We need to show that
\begin{align*}
\E_{\vec{t}_{J}}[u_i(t_i,M(\vec{t}'_C, \vec{t}'_I, \vec{t}_{J}))] \leq \E_{\vec{t}_J}[u_i(t_i,M(\vec{t}_C, \vec{t}'_I, \vec{t}_{J}))] + (r\epsilon + 2r\delta)(2\alpha). \tag{2}
\end{align*}
Following \cite{DRV10}, given two discrete random variables $X$ and $Y$, let
\begin{align*}
D_{\infty}^{\delta'}(X||Y) := \max_{W \subseteq Supp(X) \text{ s.t. } Pr[X \in W] > \delta'} \ln\left(\frac{\Pr[X \in W]-\delta'}{\Pr[Y \in W]}\right)
\end{align*}
and
\begin{align*}
\Delta(X,Y) := \ & \max_{W \subseteq Supp(X) \cup Supp(Y)} |\Pr[X \in W] - \Pr[Y \in W]| \\
= \ & \frac{1}{2} \sum_{w \in Supp(X) \cup Supp(Y)} |\Pr[X = w] - \Pr[Y = w]|.
\end{align*}
The first distance measure is known as $\delta'$-approximate max-divergence (e.g., see \cite{DRV10}) and the second distance measure is known as statistical distance (also known as statistical difference or total variation distance). It is known that if $D_{\infty}^{\delta'}(X||Y) \leq \eps'$ and $D_{\infty}^{\delta'}(Y||X) \leq \eps'$, then there exists a random variable $Z$ with support $Supp(X) \cup Supp(Y)$ such that $\Delta(X,Z) \leq \delta'$, and $D_{\infty}^{0}(Z||Y) \leq \eps'$ and $D_{\infty}^{0}(Y||Z) \leq \eps'$ (see \cite{DRV10}).

Let $X = M(\vec{t}_C, \vec{t}'_I, \vec{t}_J)$ and $X' = M(\vec{t}_C', \vec{t}'_I, \vec{t}_J)$. Then, by (1) and the symmetry of $\vec{t}_C$ and $\vec{t}_C'$, we have $D_{\infty}^{2r\delta}(X||X') \leq r \eps$ and $D_{\infty}^{2r\delta}(X'||X) \leq r \eps$. Then, there exists a random variable $Z$ such that $\Delta(X',Z) \leq 2r\delta$ and $D_{\infty}^{0}(Z||X) \leq r \eps$. From $D_{\infty}^{0}(Z||X) \leq r \eps$, we have that for every $s \in S$,
\begin{align*}
\Pr[Z = s] \leq e^{r\epsilon} \Pr[X = s]. \tag{3}
\end{align*}

\begin{claim}
$\E[u_i(t_i,X')] \leq \E[u_i(t_i,Z)] + 4r\delta\alpha$
\end{claim}

\begin{proof}[Proof of claim]
Observe that
\begin{align*}
|\E[u_i(t_i,X')] - \E[u_i(t_i,Z)]| 
= \ & |\sum_{s \in S} \Pr[X'=s] \cdot u_i(t_i,s) - \sum_{s \in S} \Pr[Z=s] \cdot u_i(t_i,s)]| \\
\leq \ & \sum_{s \in S} |\Pr[X'=s] - \Pr[Z=s]| \cdot |u_i(t_i,s)| \\
\leq \ & (2\Delta(X',Z)) \cdot \alpha \\
\leq \ & 4r\delta\alpha.
\end{align*}
\end{proof}

\begin{claim}
$\E[u_i(t_i,Z)] \leq e^{r\eps} \cdot \E[u_i(t_i,X)]$
\end{claim}

\begin{proof}[Proof of claim]
Observe that
\begin{align*}
\E[u_i(t_i,Z)]
= \sum_{s \in S} \Pr[Z=s] \cdot u_i(t_i,s) 
\leq \sum_{s \in S} e^{r\eps} \Pr[X=s] \cdot u_i(t_i,s)
= e^{r\eps} \cdot \E[u_i(t_i,X)],
\end{align*}
where the inequality follows from (3).
\end{proof}

Now, using the above two claims, we have
\begin{align*}
\E[u_i(t_i,M(\vec{t}_C',\vec{t}_I',\vec{t}_J))] 
= \ & \E[u_i(t_i,X')] \\
\leq \ & \E[u_i(t_i,Z)] + 4r\delta\alpha \\
\leq \ & e^{r\eps} \cdot \E[u_i(t_i,X)] + 4r\delta\alpha \\
= \ & e^{r\eps} \cdot \E[u_i(t_i,M(\vec{t}_C,\vec{t}_I',\vec{t}_J))] + 4r\delta\alpha \\
= \ & \E[u_i(t_i,M(\vec{t}_C,\vec{t}_I',\vec{t}_J))] + (e^{r\eps}-1) \cdot \E[u_i(t_i,M(\vec{t}_C,\vec{t}_I',\vec{t}_J))] + 4r\delta\alpha \\
\leq \ & \E[u_i(t_i,M(\vec{t}_C,\vec{t}_I',\vec{t}_J))] + 2r\eps\alpha + 4r\delta\alpha \\
= \ & \E[u_i(t_i,M(\vec{t}_C,\vec{t}_I',\vec{t}_J))] + (r\eps + 2r\delta)(2\alpha).
\end{align*}
where the last inequality follows from the fact that $e^{r\eps}-1 \leq 2r\eps$ for $r\eps < 1$. This completes the proof.
\end{proof}

\end{document}